\documentclass[10pt]{article}    
\usepackage[utf8]{inputenc}                

\usepackage[margin=1in]{geometry} 
\usepackage{changepage}
\usepackage{pdflscape}

\usepackage{floatrow}
\usepackage[labelformat=simple]{caption, subfig}
\usepackage{mathtools}
\usepackage{amsthm,amssymb}        
\usepackage{mathabx}
\usepackage{subfiles}
\usepackage{graphbox}
\usepackage{adjustbox}
\usepackage{epstopdf}
\usepackage{array}
\usepackage{breqn}
\usepackage{enumitem}
\usepackage{multirow} 
\usepackage{longtable}
\usepackage{supertabular}
\usepackage{booktabs}

\usepackage[hidelinks]{hyperref}
\usepackage{xcolor}
\hypersetup{
    colorlinks,
    linkcolor={red!50!black},
    citecolor={blue!50!black},
    urlcolor={blue!80!black}
}

\usepackage[tracking]{microtype}
\usepackage{natbib}

\newcommand*{\cp}[2]{#1 \,\square\, #2}    
\newcommand{\house}{\,\widehat{\square}\,}  
\newcommand*{\cphat}[2]{#1 \house #2}

\newtheorem{theorem}{Theorem}
\newtheorem{proposition}[theorem]{Proposition}
\newtheorem{corollary}[theorem]{Corollary}

\newtheorem{definition}[theorem]{Definition}%[section]

\newtheorem{example}[theorem]{Example}

\title{A lattice structure for ancestral configurations arising from the relationship between gene trees and species trees}
\author{Egor Lappo$^*$ and Noah A.\ Rosenberg\thanks{Department of Biology, Stanford University, Stanford, CA 94305, USA}}
\date{\today}

\begin{document}

\maketitle

%%%%%%%%%%%%%%%%%%%%%%%%%%%%%%%%%%%%%%%%%%%%%%%%%%%%%%%%%%
\noindent {\bf Abstract.} To a given gene tree topology $G$ and species tree topology $S$ with leaves labeled bijectively from a fixed set $X$, one can associate a set of \emph{ancestral configurations}, each of which encodes a set of gene lineages that can be found at a given node of the species tree. We introduce a lattice structure on ancestral configurations, studying the directed graphs that provide graphical representations of lattices of ancestral configurations. For a matching gene tree topology and species tree topology $G=S$, we present a method for defining the digraph of ancestral configurations from the tree topology by using iterated cartesian products of graphs. We show that a specific set of paths on the digraph of ancestral configurations is in bijection with the set of \emph{labeled histories} --- a well-known phylogenetic object that enumerates possible temporal orderings of the coalescences of a tree. For each of a series of tree families, we obtain closed-form expressions for the number of labeled histories by using this bijection to count paths on associated digraphs. Finally, we prove that our lattice construction extends to nonmatching tree pairs, and we use it to characterize pairs $(G,S)$ having the maximal number of ancestral configurations for a fixed $G$. We discuss how the construction provides new methods for performing enumerations of combinatorial aspects of gene and species trees.

\vskip .3cm


\vskip .15cm


\clearpage
%%%%%%%%%%%%%%%%%%%%%%%% Section 1 %%%%%%%%%%%%%%%%%%%%%%%%%%%%%%%%
\section{Introduction}

A central problem in phylogenetics is the inference of a species tree, representing evolutionary descent of species, from gene trees that describe the evolution of genealogical lineages within species. To perform this inference, procedures for estimating species trees often make use of computations of a likelihood function for species trees given information on gene trees. The  computation is performed by evaluating a probabilistic expression over each element of a set encoding purely combinatorial relations between gene tree and species tree topologies. These sets can in turn be studied mathematically to advance the theoretical understanding of species tree estimation, and of the relationships of gene trees and species trees more generally.

The computation of probabilities of gene tree topologies conditional on species trees can make use of any of a rich variety of combinatorial structures that describe pairwise relationships between tree topologies~\citep{degnan2005gene, rosenberg2007counting, degnanRosenbergStadler2012, degnanStadler2012polynomial, wu2012coalescent, wu2016cpct}. One type of structure that has been shown to be promising for such computations is the \emph{ancestral configuration}~\citep{wu2012coalescent}. Consider a species tree and a gene tree that evolves along the branches of that species tree. For each node of the species tree topology, an ancestral configuration lists the lineages of the gene tree present at a point right before that node, when viewed backward in time.

An algorithm of \cite{wu2012coalescent} performs a probability calculation over sets of ancestral configurations possible for a given gene tree topology and species tree topology. Mathematical properties of ancestral configurations then inform the calculation. For a matching gene tree and species tree topology, \cite{disanto2017enumeration} obtained enumerative results on the number of root ancestral configurations, or ancestral configurations present at the species tree root. They showed that the number of ancestral configurations present at a species tree node can be obtained recursively from corresponding values at its immediate descendant nodes. They also obtained a variety of results concerning exponential growth in the tree size of the number of ancestral configurations. Further features of this exponential growth have been obtained by \cite{disanto2022distributions} and \cite{disanto2024distributions}.

In this paper, we introduce a new way to view ancestral configurations that is based on introducing a lattice structure on the set of ancestral confgurations at a given node. We show that the set of root ancestral configurations for a given pair of trees possesses a rich algebraic structure that mirrors the recursive nature of binary trees. The lattice formulation enables us to prove novel enumerative results connecting ancestral configurations with another set of structures, the labeled histories. 

In Section~\ref{sec:def}, we give the necessary biological and mathematical preliminaries. In Section~\ref{sec:construction}, we define a lattice structure on the set of ancestral configurations for matching gene trees and species trees, and in Section~\ref{sec:diagrams}, we introduce a visual representation of lattices of ancestral configurations in the form of Hasse diagrams. In Section~\ref{sec:chains}, we prove a bijection between the set of maximal chains in the lattice of ancestral configurations for a given tree and its set of labeled histories. Section~\ref{sec:examples} shows examples of Hasse diagrams associated with ancestral configurations for all unlabeled tree topologies with eight or fewer leaves. In Section~\ref{sec:structure_description}, we prove a recurrence relation for Hasse diagrams analogous to that for sets of ancestral configurations. Applications of the recurrence to specific families of trees appear in Section~\ref{sec:families}. Finally, in Section~\ref{sec:nonmatching}, we extend the lattice structure to include nonmatching tree pairs. We conclude with a discussion in Section \ref{sec:discussion}.

%%%%%%%%%%%%%%%%%%%%%%%% Section 2 %%%%%%%%%%%%%%%%%%%%%%%%%%%%%%%%%%%
\section{Definitions}\label{sec:def}

We focus on leaf-labeled binary rooted trees. In this paper, a \emph{tree} is assumed to mean a binary rooted tree with leaves labeled uniquely by elements of a set of labels $X = \{a, b, c, \ldots\}$. 

In our definitions of ancestral configurations, we closely follow \cite{disanto2017enumeration}; see \cite{steel2016phylogeny} for general terminology. Consider a gene tree $G$ and species tree $S$, both of which represent labeled topologies. For a given gene tree $G$ and species tree $S$ with labels in bijective correspondence, a \emph{realization} $R$ of $G$ in $S$ is one of the evolutionary possibilities of $G$ appearing on $S$. If we consider the edges of $S$ to be populations, then $R$ specifies the lineages, or edges, of $G$ that are present in each of the edges of $S$.

For a given node $v$ of $S$, we let $C(v, G, S, R)$ denote a set of lineages of $G$ present in $S$ at a point right before $v$, when viewed backward in time. We say that $C(v,G,S,R)$ is an \emph{ancestral configuration of $G$ at a node $v$ of $S$}. Denote by $\mathcal R(G,S)$ the set of \emph{all realizations} of $G$ in $S$. We can define a set 
\[
    C(v, G, S) = \{C(v, G, S, R) \mid R \in \mathcal R(G,S)\}
\]
of \emph{ancestral configurations} at node $v$. The elements of $C(v, G, S)$ enumerate the ways that lineages of $G$ can reach the point right before node $v$ in $S$, viewed backward in time. Figure~\ref{fig:realizations} presents example realizations of $G$ in $S$ when trees have matching topologies.

%%%%%%%%%%%%%%%%%%%%%%%%%%%% Figure 1 %%%%%%%%%%%%%%%%%%%%%%%%%%%%%%%%
\begin{figure}
    \centering
    \begin{minipage}{0.44\textwidth}
        \centering
        \includegraphics[scale=0.38]{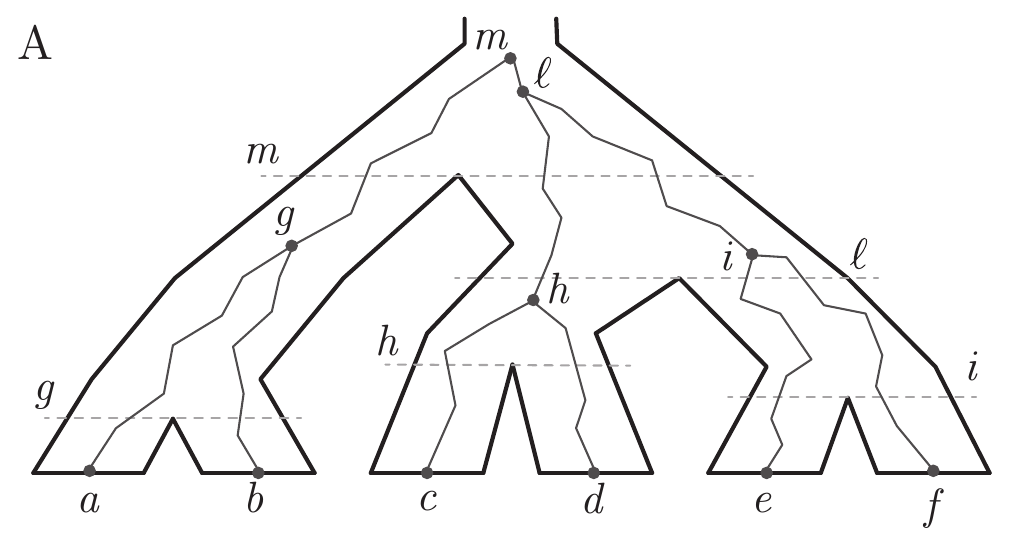}
    \end{minipage}\hfill
    \begin{minipage}{0.44\textwidth}
        \centering
        \includegraphics[scale=0.38]{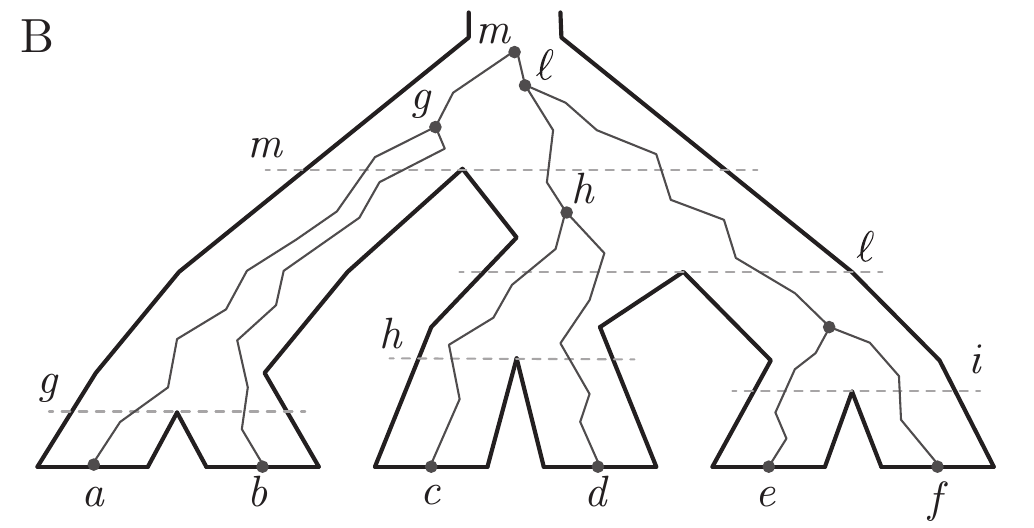}
    \end{minipage}
    \caption{Two example realizations of a gene tree $G$ in a species tree $S$ with matching topology. We identify each lineage of $G$ by its immediate descendant node. The internal nodes of the species tree $S$ are represented by horizontal dashed lines. (A) In this realization, the ancestral configuration is $\{c,d\}$ at node $h$ of $S$; $\{e,f\}$ at node $i$ of $S$; $\{a,b\}$ at node $g$ of $S$; $\{h,e,f\}$ at node $\ell$ of $S$; $\{g,h,i\}$ at node $m$ of $S$. (B) In this realization, the ancestral configuration is $\{c,d\}$ at node $h$ of $S$; $\{e,f\}$ at node $i$ of $S$; $\{a,b\}$ at node $g$ of $S$; $\{c,d,i\}$ at node $\ell$ of $S$; $\{a,b,h,i\}$ at node $m$ of $S$.}
    \label{fig:realizations} 
\end{figure}
%%%%%%%%%%%%%%%%%%%%%%%%%%%%%%%%%%%%%%%%%%%%%%%%%%%%%%%%%%%%%%%%%%%%

Of particular interest are \emph{root ancestral configurations}, for which $v$ is the root node of $S$. \cite{disanto2017enumeration} showed that knowing the number of root ancestral configurations of a gene tree $G$ enables a calculation of an upper bound on the sum across nodes $v$ of $S$ of its numbers of ancestral configurations at $v$. We denote the set of root ancestral configurations for a pair $(G,S)$ by $C(G,S)$, and its cardinality by $c(G, S) \equiv |C(G,S)|$. We often refer to root ancestral configurations as \emph{root configurations} for short, and to ancestral configurations simply as \emph{configurations}.

Finally, in most of the paper,  we restrict ourselves to the case of equal topologies for a gene and species tree, $G=S$; in this case, we denote $C(S,S)$ and $c(S,S)$ by $C(S)$ and $c(S)$, respectively. After we present our results for tree pairs with matching topologies, we extend them to nonmatching pairs ($G\neq S$) in Section~\ref{sec:nonmatching}.

%%%%%%%%%%%%%%%%%%%%%%%% Section 3 %%%%%%%%%%%%%%%%%%%%%%%%%%%%%%
\section{Lattice of ancestral configurations}\label{sec:construction}

Let $S$ be a tree with $n$ leaves. We introduce the structure of a partially-ordered set, or poset, on the set of root ancestral configurations $C(S)$ for a pair consisting of an identical, or matching, gene tree topology and species tree topology. 

Recall that a \emph{poset} $P$ is a set together with a binary relation $\prec$ that satisfies, for all $a,b,c \in P$:
	\begin{enumerate}[label=(\roman*)]
		\item $a \prec a$;
		\item If $a \prec b$ and $b \prec a$ then $a = b$;
		\item If $a \prec b$ and $b \prec c$ then $a \prec c$.
	\end{enumerate}
The relation $\prec$ is a \emph{partial order}. A poset $(P,\prec)$ is termed  a \emph{lattice} if each pair of elements of $P$ has both a greatest lower bound and a least upper bound with respect to $\prec$. A lattice is \emph{bounded} if it has both a maximal and a minimal element \citep{gratzer}.

In the case of ancestral configurations, for $a,b \in C(S)$ we introduce a relation $\prec$. We say $a \prec b$ if by coalescing some lineages (possibly none) from the ancestral configuration $a$, we can obtain the ancestral configuration $b$, as illustrated in Figure~\ref{fig:poset_relation}. Verification of the three properties (i), (ii), and (iii)  is immediate, so that $(C(S),\prec)$ is a poset.

%%%%%%%%%%%%%%%%%%%%%%%%%%% Figure 2 %%%%%%%%%%%%%%%%%%%%%%%%%%%%%%%%%%%%
\begin{figure}
    \centering
    \hspace{1cm}\includegraphics[width=0.65\textwidth]{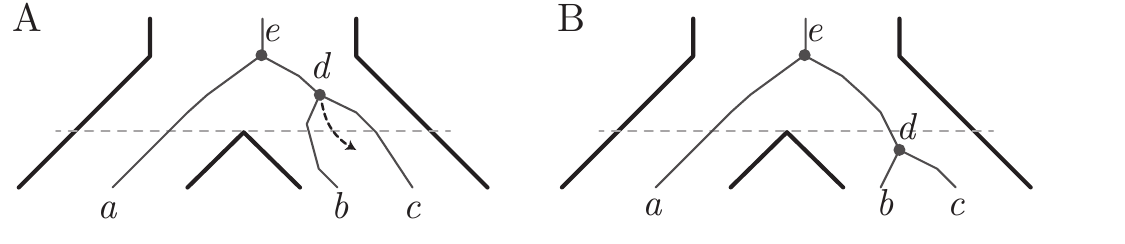}
    \caption{A visual representation of the partial order relation on ancestral configurations. (A) At the internal node of the species tree represented by the dashed line, the ancestral configuration is $\{a,b,c\}$. (B) At the internal node of the species tree, the ancestral configuration is $\{a,d\}$. By the partial order relation, $\{a,b,c\} \prec \{a,d\}$. Comparing (A) and (B), coalescing lineages of the ancestral configuration corresponds to moving a gene tree node, $d$ in this case, across a species tree node.}
    \label{fig:poset_relation}
\end{figure}
%%%%%%%%%%%%%%%%%%%%%%%%%%%%%%%%%%%%%%%%%%%%%%%%%%%%%%%%%%%%%%%%%%%%%%%%%%%%%

In fact, $C(S)$ has an even richer structure: it is a \emph{bounded lattice}, possessing both a least element and a greatest element. A least upper bound for $a,b \in C$ is obtained by separately coalescing lineages in $a$ and lineages in $b$ until both procedures reach the same configuration, as shown in Figure~\ref{fig:poset_relation}. This process of repeated coalescence necessarily terminates, giving a root ancestral configuration containing two lineages---the immediate descendant lineages of the root. The resulting configuration, which we can denote $c_{\max}$, satisfies $a \prec c_{\max}$ and $b \prec c_{\max}$, and hence is an upper bound for $a,b$. Because the sets $A = \{c \in C\mid a \prec c\}$ and $B = \{c \in C\mid b \prec c\}$ are finite, $A \cap B$ has a minimal element, which is the required least upper bound. The proof of existence of an infimum follows a similar argument, noting that the configuration $c_{\min}$ containing all leaf nodes satisfies $c_{\min} \prec c$ for all $c \in C$.

We have therefore demonstrated the following proposition.

\begin{proposition}\label{prop:acs_are_lattices}
	The set $C(S)$ of root ancestral configurations for the matching pair of labeled topologies $(G,S)$ is a bounded lattice.
\end{proposition}

We illustrate the proposition with an example. Consider the labeled topology $S = ((a,b)_g,((c,d)_h,(e,f)_i)_\ell)_m$ on $6$ leaves in Figure~\ref{fig:realizations}, where the subscripts label the nodes. The set $C(S)$ of root ancestral configurations is
\[
	C = \left\{\begin{array}{c}
		\{g,\ell\}, \{a,b,\ell\}, \{g,c,d,e,f\}, \{g,c,d,i\}, \\
		\{a,b,c,d,e,f\}, \{g,h,e,f\}, \{a,b,h,e,f\}, \\
		\{a,b,c,d,i\}, \{g,h,i\}, \{a,b,h,i\}
	\end{array}\right\}.
\]
In particular, $\{a,b,c,d,e,f\}$ is the minimal element and $\{g,\ell\}$ is the maximal element. By coalescing lineage pairs, we see that we have relations, $\{a,b,h,i\}\prec\{a,b,\ell\}$, $\{g,c,d,e,f\}\prec\{g,c,d,i\}$, and $\{g,c,d,e,f\}\prec\{g,h,e,f\}$, among others.

Now that we have defined a lattice algebraic structure on $C(S)$, we can use ideas of lattice theory to prove results about ancestral configurations. We next focus on representations of lattices by diagrams or graphs.

%%%%%%%%%%%%%%%%%%%%%%%% Section 4 %%%%%%%%%%%%%%%%%%%%%%%%%%%%%%%%%%%%%%%%
\section{Diagrams associated to lattices of ancestral configurations}\label{sec:diagrams}

For a lattice $L$, one can associate a directed graph, or digraph, $D(L)$. This digraph $D(L)$ has a node for each element $x$ in $L$, and an edge from $x$ to $y$ if and only if $y$ \emph{covers} $x$---that is, if and only if (i) $x \prec y$, (ii) no $z$ exists such that $x \prec z \prec y$. A visual representation of $D(L)$ with all of the directed edges oriented in the same direction (say, upward) is the \emph{Hasse diagram} of $L$ \citep{gratzer}. 

For convenience, we denote simply by $D(S)$ the digraph $D\big( C(S) \big)$ associated with the lattice of root ancestral configurations of the species tree $S$. To illustrate this concept, we consider two examples. Figure~\ref{fig:ex1} shows the labeled topology $S = ((a,b),c),(d,e))$ and a corresponding Hasse diagram for its set $C(S)$ of root configurations. We see that $S$ has six root configurations, with maximal element $\{g,h\}$ and minimal element $\{a,b,c,d,e\}$. An edge is drawn only between pairs of ancestral configurations that differ by substituting two lineages with one coalesced lineage. We sometimes use $D(S)$ to refer to both a digraph and its visual representation as a Hasse diagram. 

Figure~\ref{fig:ex2} shows the balanced labeled topology $S = (((a,b),(c,d)),((e,f),(g,h)))$ and its corresponding Hasse diagram. We see that $S$ has $25$ root configurations; like $S$, the Hasse diagram is symmetric. The subsets of $C(S)$ with 8, 7, 6, 5, 4, 3, and 2 elements have sizes 
%$1$, $2$, $5$, $6$, $6$, $4$, and $1$, respectively.
1, 4, 6, 6, 5, 2, and 1, respectively.

We note that the Hasse diagram for the set $C(S)$ of root ancestral configurations subsumes the ancestral configurations at each of the nodes of $S$. Let $v \in S$ be an internal node of $S$, and let $S|_v$ be the subtree of $S$ rooted at $v$. We have a natural inclusion $C(S|_v) \hookrightarrow C(S)$ that takes an ancestral configuration $x$ in $C(S|_v)$ and associates to it an ancestral configuration $x \cup \{\text{leaves of $S$ that are not in $S|_v$}\}$. The Hasse diagram for $C(S|_v)$ can then be viewed as contained in the Hasse diagram for $C(S)$. 

%%%%%%%%%%%%%%%%%%%%%%%%%%% Figure 3 %%%%%%%%%%%%%%%%%%%%%%%%%%%%%%%%%%%%%%%%
\begin{figure}[tb]\centering
    \hspace{1.6cm}\begin{minipage}[b]{0.45\textwidth} A 
            
            {\centering
        \includegraphics[height=1.5in]{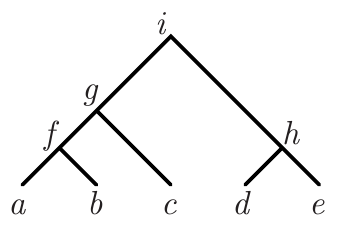}}
    \end{minipage}\hfill
    \begin{minipage}[b]{0.45\textwidth} B 
            
            {\centering
        \includegraphics[height=1.5in]{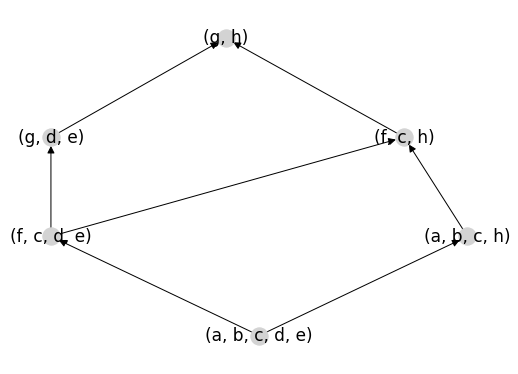}}
    \end{minipage}
\caption{A labeled topology on 5 leaves and its associated Hasse diagram. (A) A labeled topology $S$ on $5$ leaves. (B) The corresponding Hasse diagram for $D(S)$.}
\label{fig:ex1}
\end{figure}
%%%%%%%%%%%%%%%%%%%%%%%%%%%%%%%%%%%%%%%%%%%%%%%%%%%%%%%%%%%%%%%%%%%%%%%%%%%%%

%%%%%%%%%%%%%%%%%%%%%%%%%%% Figure 4 %%%%%%%%%%%%%%%%%%%%%%%%%%%%%%%%%%%%%%%
\begin{figure}[tb]\centering
    \hspace{-0.5cm}\begin{minipage}[b]{0.45\textwidth} A 
            
            {\centering
        \includegraphics[height=1.5in]{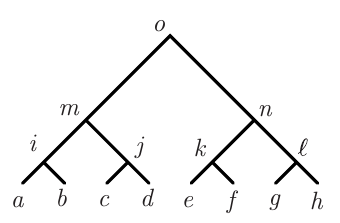}}
    \end{minipage}\hspace{-1cm}
    \begin{minipage}[b]{0.45\textwidth} B 
            
            {\centering
        \includegraphics[height=1.5in]{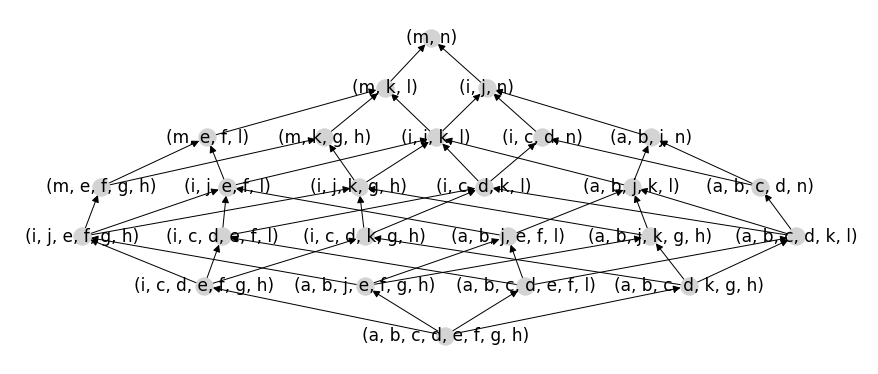}}
    \end{minipage}
\caption{A balanced labeled topology on 8 leaves and its associated Hasse diagram. (A) A balanced labeled topology $S$ on $8$ leaves. (B) The corresponding Hasse diagram for $D(S)$.}
\label{fig:ex2}
\end{figure}
%%%%%%%%%%%%%%%%%%%%%%%%%%%%%%%%%%%%%%%%%%%%%%%%%%%%%%%%%%%%%%%%%%%%%%%%%%%%

%%%%%%%%%%%%%%%%%%%%%%% Section 5 %%%%%%%%%%%%%%%%%%%%%%%%%%%%%%%%%%%%%%%%
\section{Maximal chains and labeled histories}\label{sec:chains}

The bounded lattice structure for set $C(S)$ suggests an enumerative question: what are the \emph{maximal chains} in $C(S)$? For elements $x$, $y$, $z$ of a partial order with $x \neq y$, we say that $y$ \emph{covers} $x$ if $x \prec z \prec y$ implies $z = x \text{ or } z = y$. The maximal chains in lattice $C(S)$ are the sequences $x_1, x_2, \ldots, x_k$ of elements of $C(S)$, such that $x_1$ is the minimal element, $x_k$ is the maximal element, and for each $i$ from 1 to $k-1$, $x_{i+1}$ covers $x_i$.

Each successive entry in sequence $x_1, x_2, \ldots, x_k$ represents a coalescence of two nodes in the previous entry. For a labeled topology with $n$ leaves, $n-1$ coalescences take place. The minimal element of a maximal chain in $C(S)$ has $n$ nodes, and the maximal element has 2 nodes, so that maximal chains have length $n-1$.

Geometrically, we can reinterpret this question in terms of diagrams, noting that the maximal chains of a lattice correspond to the paths from the minimal to the maximal element of its associated Hasse diagram. We trivially obtain the following proposition.
\begin{proposition}\label{prop:equivalence}
	The enumeration of maximal chains in the lattice $C(S)$ is equivalent to the enumeration of paths in the digraph $D(S)$ that start at the minimal element and end at the maximal element. 
\end{proposition}
Consider the labeled topology $S$ from Figure~\ref{fig:ex1}A. We see that $C(S)$ has exactly three maximal chains: $\{a,b,c,d,e\} \prec \{f,c,d,e\} \prec \{g,d,e\} \prec \{g,h\}$,  $\{a,b,c,d,e\} \prec \{f,c,d,e\} \prec \{f,c,h\} \prec \{g,h\}$,  $\{a,b,c,d,e\} \prec \{a,b,c,h\} \prec \{f,c,h\} \prec \{g,h\}$. These maximal chains are in bijective correspondence with the three paths from the bottom to the top in the diagram in Figure~\ref{fig:ex1}B.

One consequence of Proposition \ref{prop:equivalence} is a connection of maximal chains of $C(S)$ with labeled histories for $S$. A \emph{labeled history} for a tree $S$ is a sequence of its coalescences, or, more precisely, a ranking of the internal nodes of $S$ such that if a node $u$ is a descendant of $v$, then $u$ is ranked higher than $v$ \citep[p.~45]{steel2016phylogeny}. 
The following theorem states the connection between the set of root ancestral configurations for a labeled topology $S$ and the set of labeled histories for $S$.
\begin{theorem}\label{thm:histories_paths}
	For a labeled topology $S$, the following three sets are in bijective correspondence: (i) the maximal chains of the lattice $C(S)$, (ii) the paths from the minimal element to the maximal element of the Hasse diagram of $D(S)$, and (iii) the labeled histories of $S$.
\end{theorem}

This statement is almost trivial, as it is simply a reformulation of the definition of labeled histories. First, the equivalence of (i) and (ii) is a restatement of Proposition \ref{prop:equivalence}. For equivalence of (ii) and (iii), note that each labeled history is a sequence of coalescences, so to each labeled history we can associate a sequence of edges of the diagram of $D(S)$: as described in Section~\ref{sec:diagrams}, an edge from $x$ to $y$ corresponds to the covering of $x$ by $y$, which, according to our definition of the poset relation in $C(S)$, involves precisely one coalescence of lineages of $S$. These edges of $D(S)$ form a path from the minimal element to the maximal element of $C(S)$; to see this, observe that if the edges do not align to form a path, then the labeled history includes coalescences of lineages that are not produced by previous coalescences --- a contradiction. Conversely, given a path in $D(S)$ from the minimal to the maximal element, we can read off a labeled history from the sequence of edges in the path, as each edge is associated to some coalescence of lineages of $S$.

Henceforth, we let $H(S)$ denote the set of labeled histories of $S$, and we let $h(S)\equiv|H(S)|$ denote the cardinality of this set. Because the paths on a digraph $D(S)$ correspond to the labeled histories of $S$ by Theorem \ref{thm:histories_paths}, we use $h(S)$ interchangeably with $h\big(D(S)\big)$ to count paths of $D(S)$. We use the same conventions for $H(S)$ and $H\big(D(S)\big)$.

The number $h(S)$ of labeled histories for a given labeled topology $S$ on $n$ lineages is equal to
\begin{equation}
\label{eq:steel_exp}
	h(S) = \frac{(n-1)!}{\prod_{r = 3}^n (r-1)^{d_r(S)}},
\end{equation}
where $d_r(S)$ is the number of internal nodes of $S$ from which exactly $r$ leaves descend \citep[][page 46]{steel2016phylogeny}. With the help of Theorem~\ref{thm:histories_paths}, this formula can be given a visual interpretation as the number of paths from the minimal element to the maximal element of the Hasse diagram of $D(S)$. 

For example, suppose $S = \mathcal{C}_n$ is a caterpillar tree on $n$ leaves. Then $\mathcal C_n$ possesses a unique linear ordering of its internal nodes. The Hasse diagram for a caterpillar tree is a linear graph with $n-1$ nodes and $n-2$ edges connecting adjacent nodes. Hence, $h(\mathcal C_n) = 1$. In eq.~\ref{eq:steel_exp}, we also see that $h(\mathcal C_n)=1$, as the factorial in the numerator is canceled in the denominator.

More generally, for any seed tree $S$ we can form a \emph{caterpillar-like family} $S_k$ by successively adjoining $k \geq 1$ caterpillar branches to the root of $S$, where a caterpillar branch is an edge that connects to a single leaf. The caterpillar branches introduce a linear ordering of the additional internal nodes ancestral to the root of the seed tree $S$, so that $h(S_k)$ does not depend on $k$ and always equals $h(S)$. Geometrically, adding caterpillar branches to $S$ corresponds to adding one edge and one vertex to the top of the Hasse diagram for $S$, an operation that does not change the number of paths from the minimal element to the maximal element.

%%%%%%%%%%%%%%%%%%%%%%%%%%%%%%%% Section 6 %%%%%%%%%%%%%%%%%%%%%%%%%%%%%
\section{Examples}\label{sec:examples}

To begin to understand the structure of lattices of ancestral configurations for labeled topologies, we provide Hasse diagrams (or digraphs $D(S)$) for all trees $S$ of size $n \leq 8$ (Tables~\ref{tbl:examples_n_3_to_7} and \ref{tbl:examples_n_8}). For each of these small trees, to construct the associated digraph, we enumerate all ancestral configurations and characterize their relationship according to $\prec$. For each tree, we list the number of ancestral configurations $c(S)$, representing the number of vertices in the digraph, and the number of labeled histories $h(S)$ (eq.~\ref{eq:steel_exp}), representing the number of paths from the minimal to the maximal element.

Informally, we can observe a number of features of the digraphs. For example, the digraphs of ancestral configurations for caterpillar trees are simply lines. For trees that possess an internal node $v$ all of whose ancestral nodes have a leaf as an immediate descendant, the digraph has a subgraph associated with $v$ and a line for the ancestors of $v$. For trees whose two immediate subtrees of the root are a caterpillar and a cherry, the digraph can be described as having two lines from the minimal to maximal element, with paths connecting these two lines. Trees for which the two immediate subtrees of the root are both caterpillars have a planar grid for their digraphs, and all monotonic paths on the grid from minimum to maximum are permissible.

We explore some of these features further in subsequent sections. 

\begin{adjustwidth}{-30pt}{30pt}

\begin{table}
\adjustbox{valign=t}{\begin{minipage}[t]{0.4\textwidth}
\begin{tabular}{ccccc}
\toprule
$n$ & Tree $S$ &$D(S)$ & $c(S)$ & $h(S)$ \\
\midrule
3 & \includegraphics[align = c,height = 0.9cm, align = c]{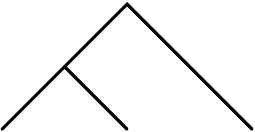} & \includegraphics[align = c,height = 1.6cm, align = c]{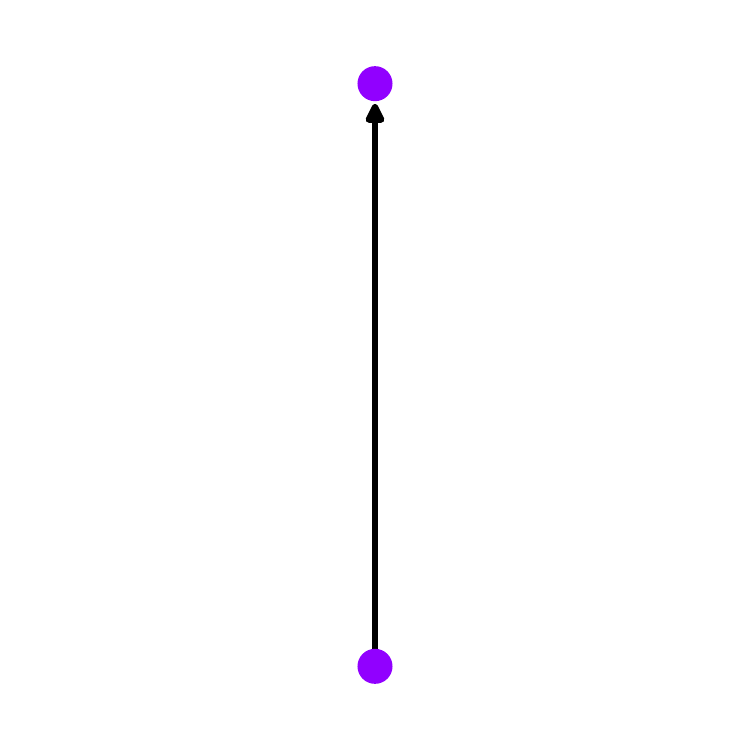} & 2 & 1 \\
4 &  \includegraphics[height = 0.9cm, align = c]{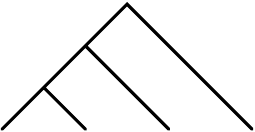} & \includegraphics[height = 1.6cm, align = c]{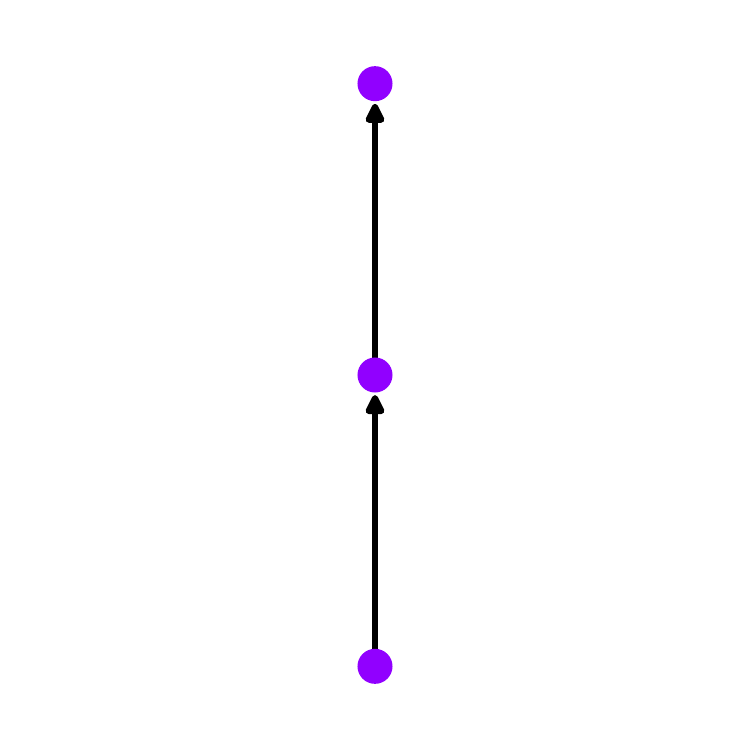} & 3 & 1 \\
 &  \includegraphics[height = 0.9cm, align = c]{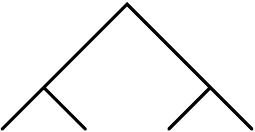} & \includegraphics[height = 1.6cm, align = c]{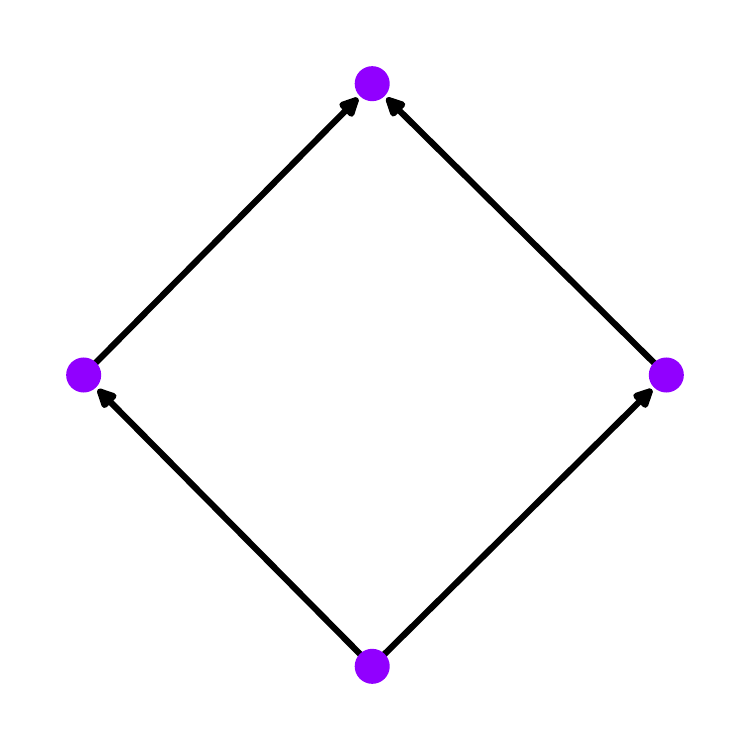} & 4 & 2 \\
5 &  \includegraphics[height = 0.9cm, align = c]{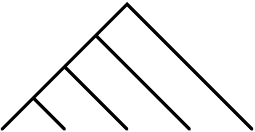} & \includegraphics[height = 1.6cm, align = c]{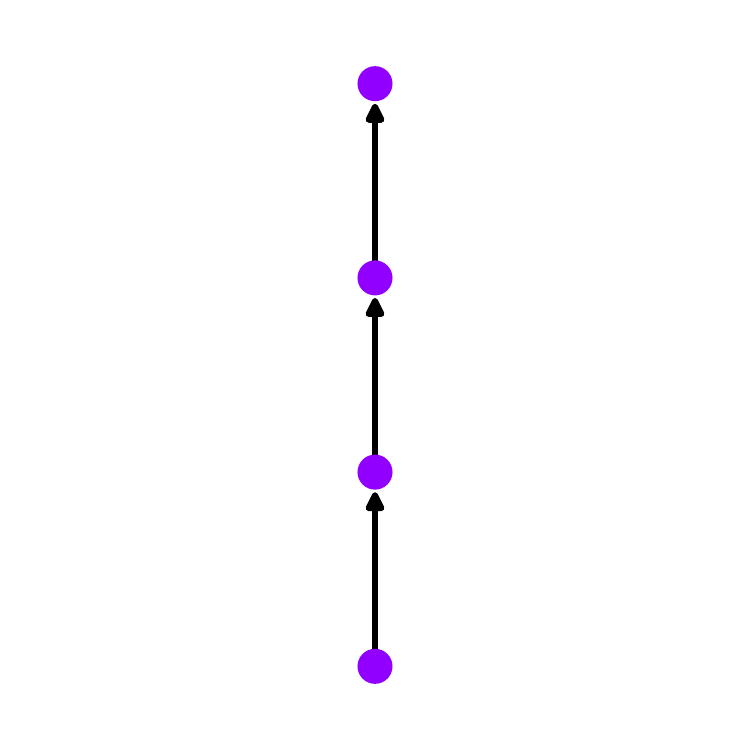} & 4 & 1 \\
 &  \includegraphics[height = 0.9cm, align = c]{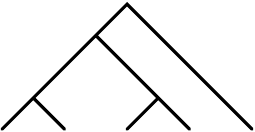} & \includegraphics[height = 1.6cm, align = c]{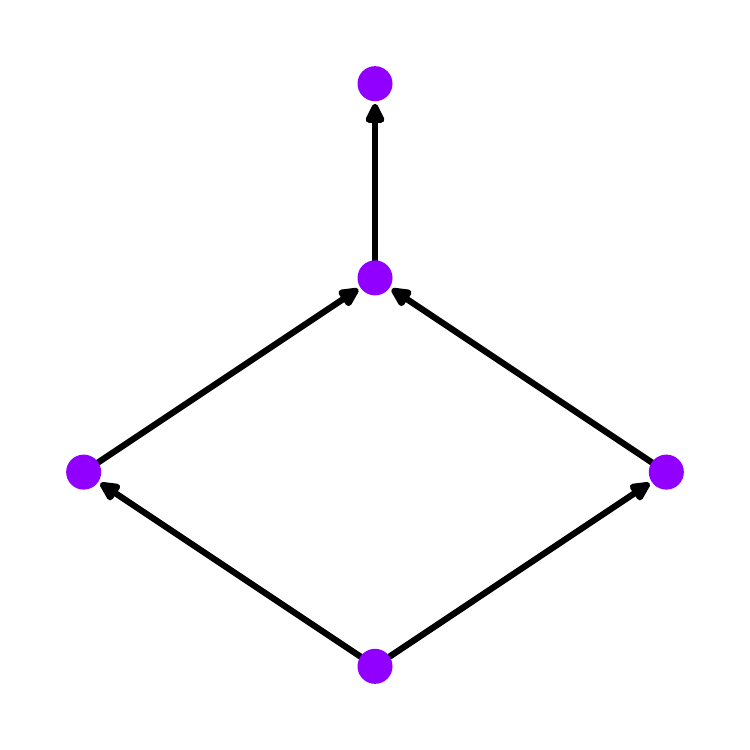} & 5 & 2 \\
 &  \includegraphics[height = 0.9cm, align = c]{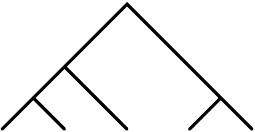} & \includegraphics[height = 1.6cm, align = c]{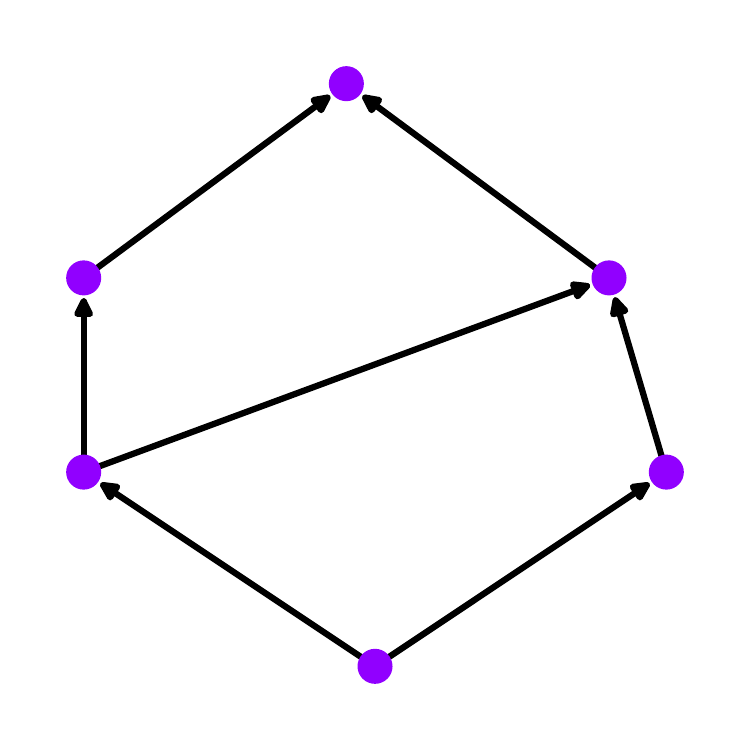} & 6 & 3 \\
6 &  \includegraphics[height = 0.9cm, align = c]{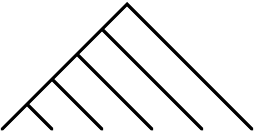} & \includegraphics[height = 1.6cm, align = c]{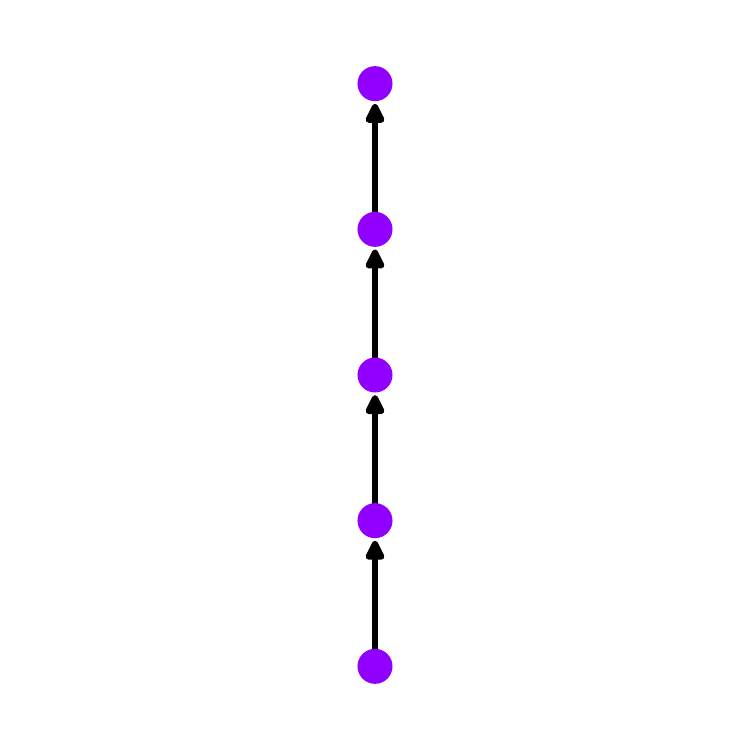} & 5 & 1 \\
 &  \includegraphics[height = 0.9cm, align = c]{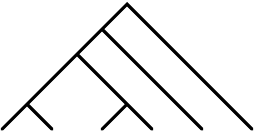} & \includegraphics[height = 1.6cm, align = c]{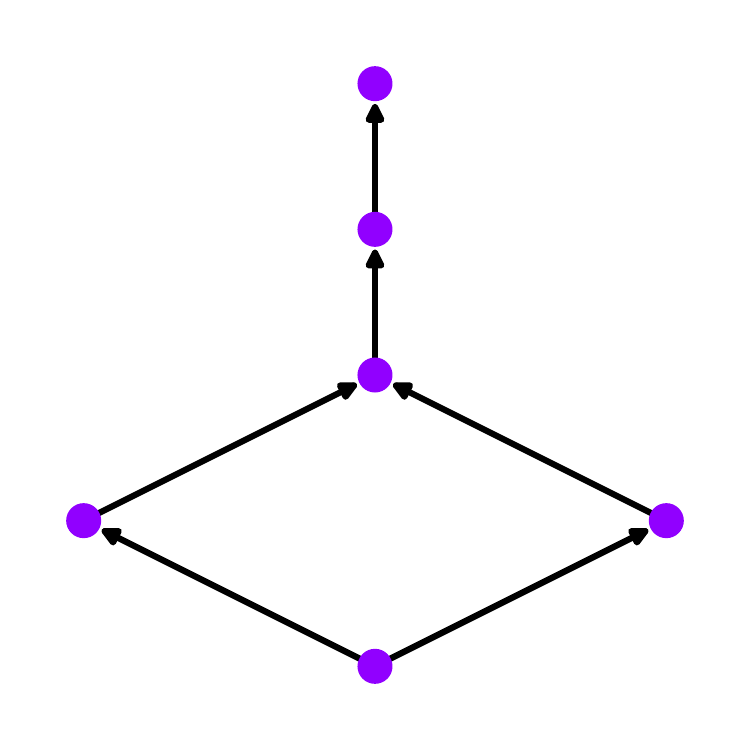} & 6 & 2 \\
 &  \includegraphics[height = 0.9cm, align = c]{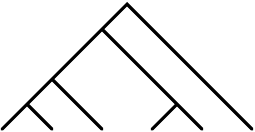} & \includegraphics[height = 1.6cm, align = c]{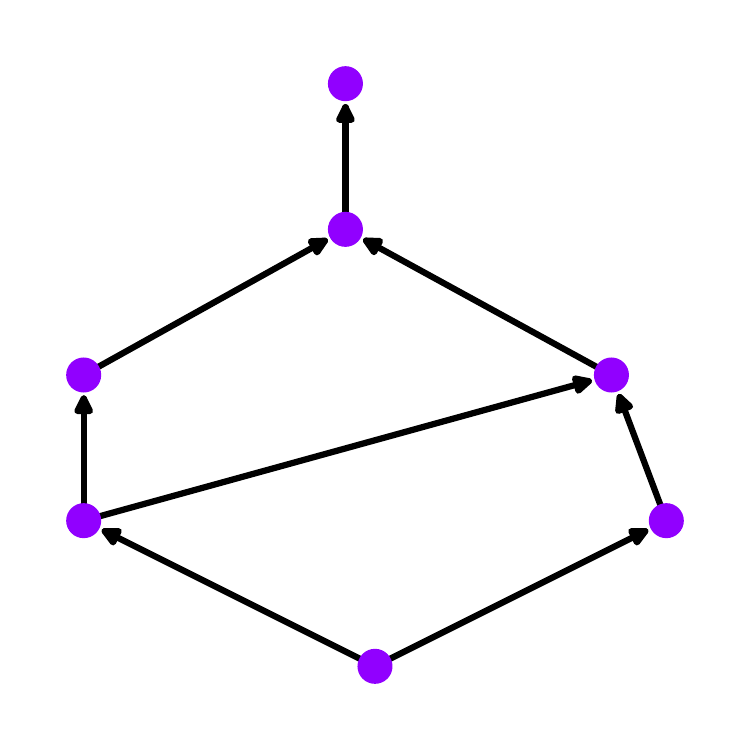} & 7 & 3 \\
  &  \includegraphics[height = 0.9cm, align = c]{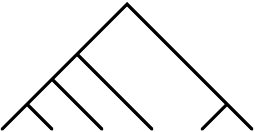} & \includegraphics[height = 1.6cm, align = c]{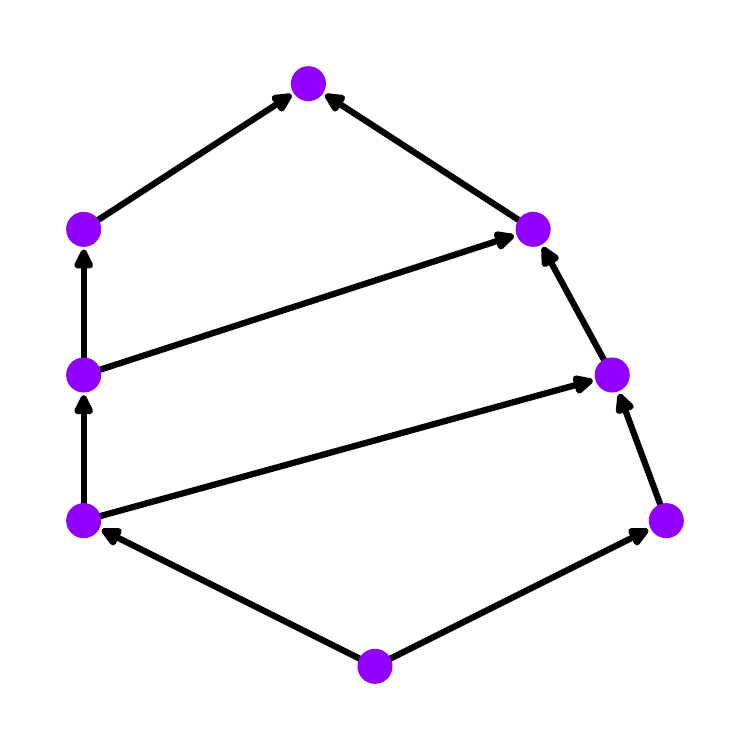} & 8 & 4 \\
 &  \includegraphics[height = 0.9cm, align = c]{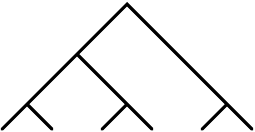} & \includegraphics[height = 1.6cm, align = c]{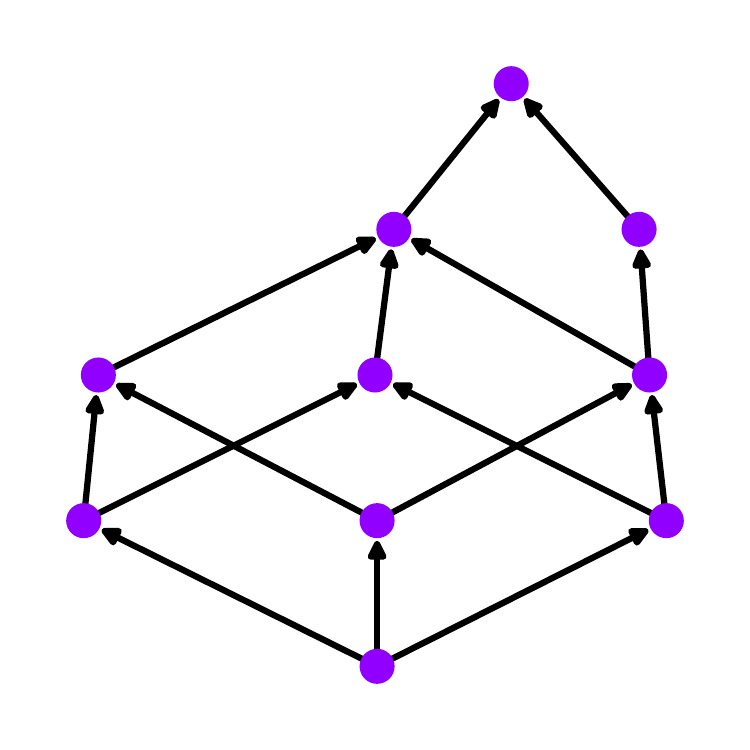} & 10 & 8 \\
 &  \includegraphics[height = 0.9cm, align = c]{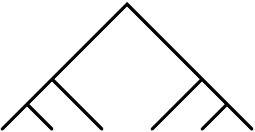} & \includegraphics[height = 1.6cm, align = c]{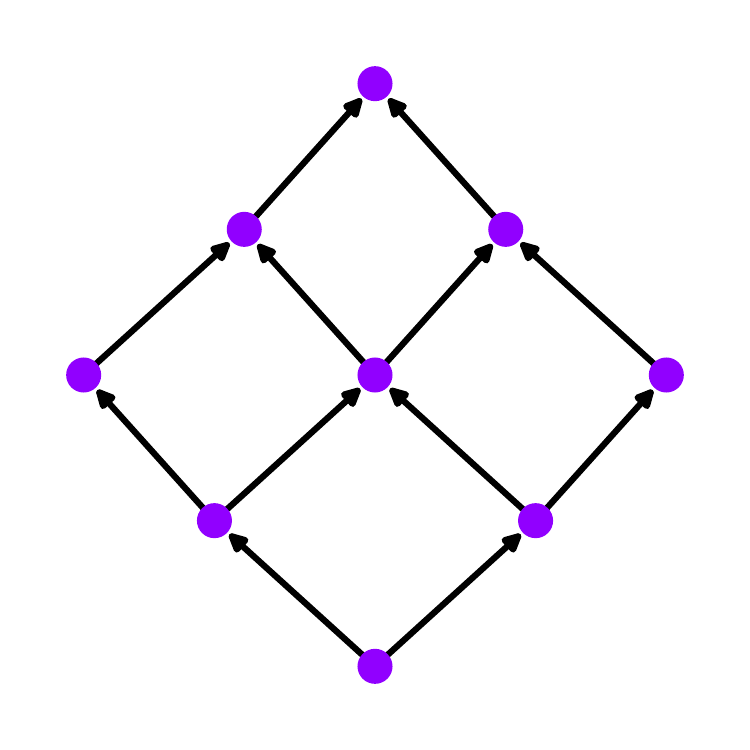} & 9 & 6 \\
\bottomrule
\end{tabular}
\end{minipage}} \hfill
\adjustbox{valign=t}{\begin{minipage}{0.4\textwidth}
\begin{tabular}{ccccc}
\toprule
$n$ & Tree $S$ &$D(S)$ & $c(S)$ & $h(S)$ \\
\midrule
7 &  \includegraphics[height = 0.9cm, align = c]{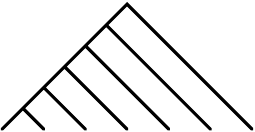} & \includegraphics[height = 1.6cm, align = c]{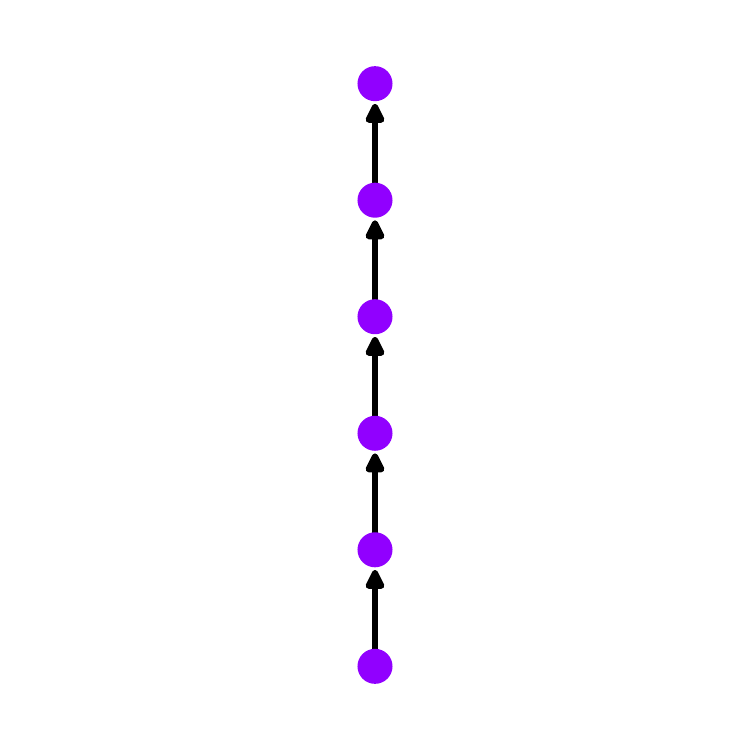} & 6 & 1 \\
 &  \includegraphics[height = 0.9cm, align = c]{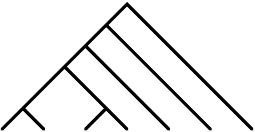} & \includegraphics[height = 1.6cm, align = c]{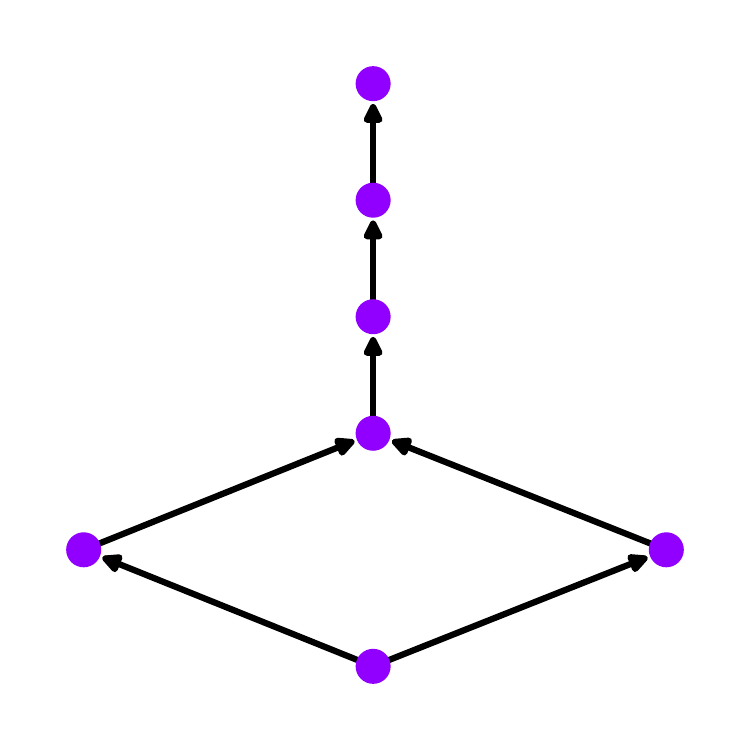} & 7 & 2 \\
 &  \includegraphics[height = 0.9cm, align = c]{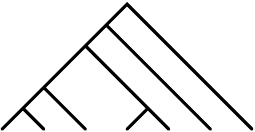} & \includegraphics[height = 1.6cm, align = c]{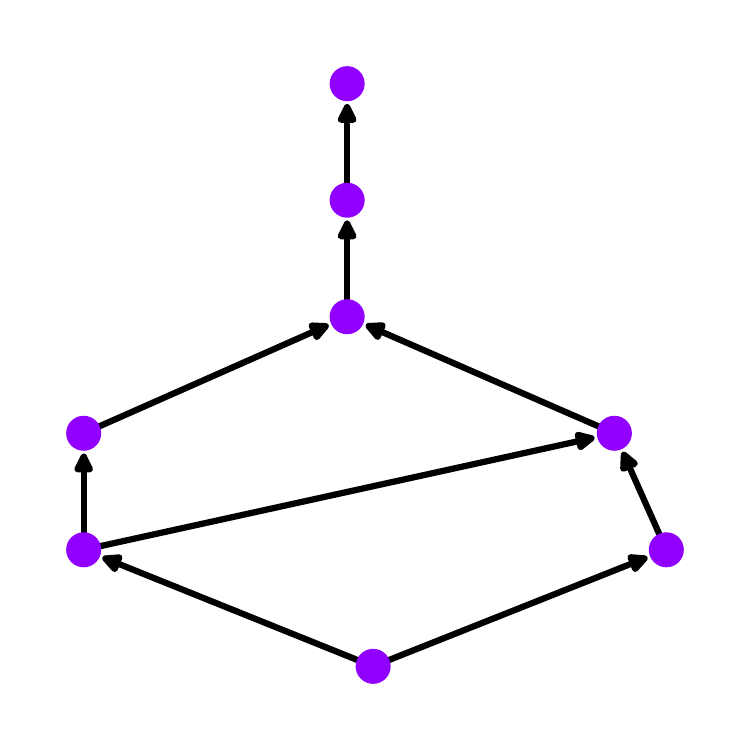} & 8 & 3 \\
 &  \includegraphics[height = 0.9cm, align = c]{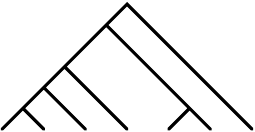} & \includegraphics[height = 1.6cm, align = c]{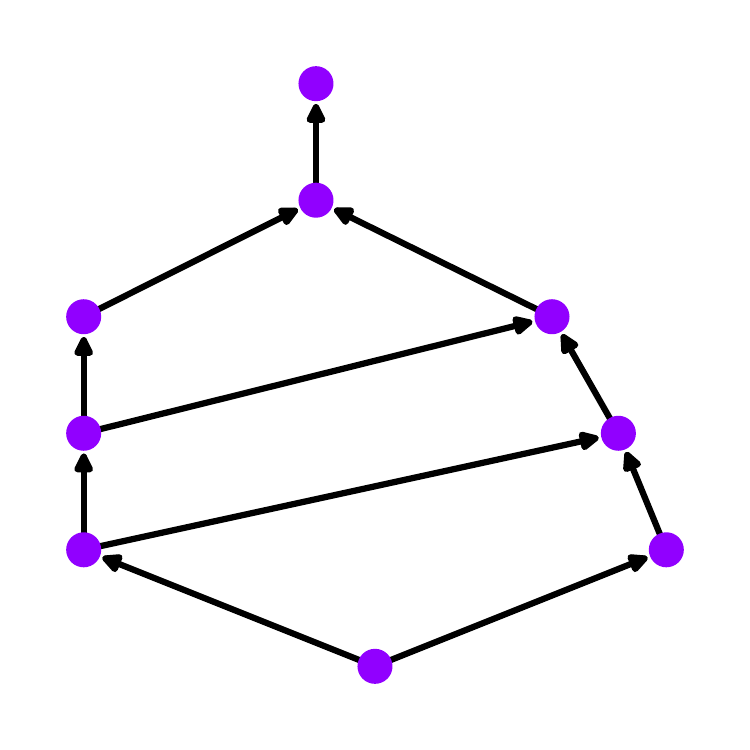} & 9 & 4 \\
 &  \includegraphics[height = 0.9cm, align = c]{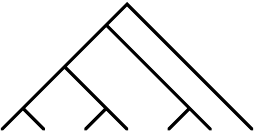} & \includegraphics[height = 1.6cm, align = c]{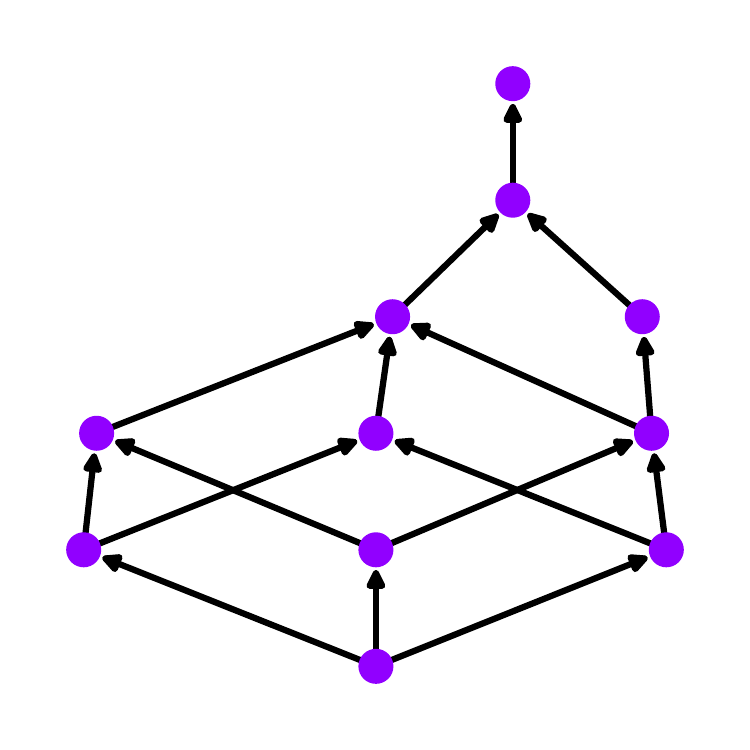} & 11 & 8 \\
 &  \includegraphics[height = 0.9cm, align = c]{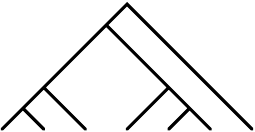} & \includegraphics[height = 1.6cm, align = c]{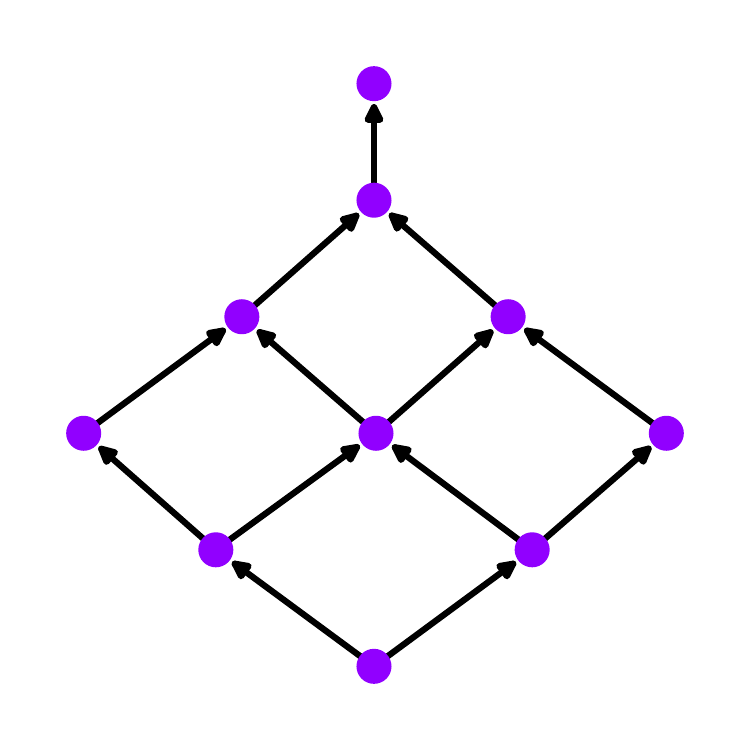} & 10 & 6 \\
 &  \includegraphics[height = 0.9cm, align = c]{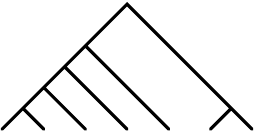} & \includegraphics[height = 1.6cm, align = c]{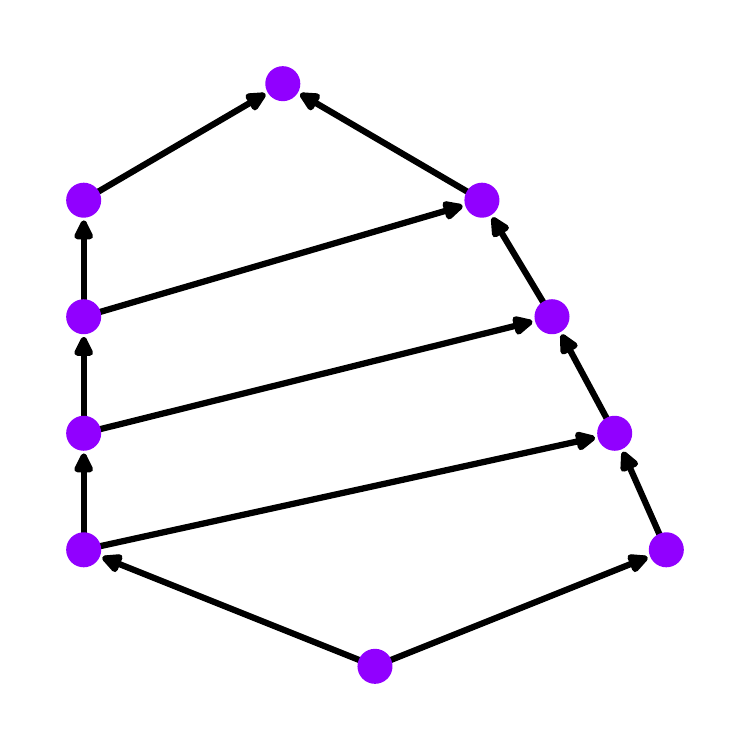} & 10 & 5 \\
 &  \includegraphics[height = 0.9cm, align = c]{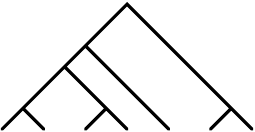} & \includegraphics[height = 1.6cm, align = c]{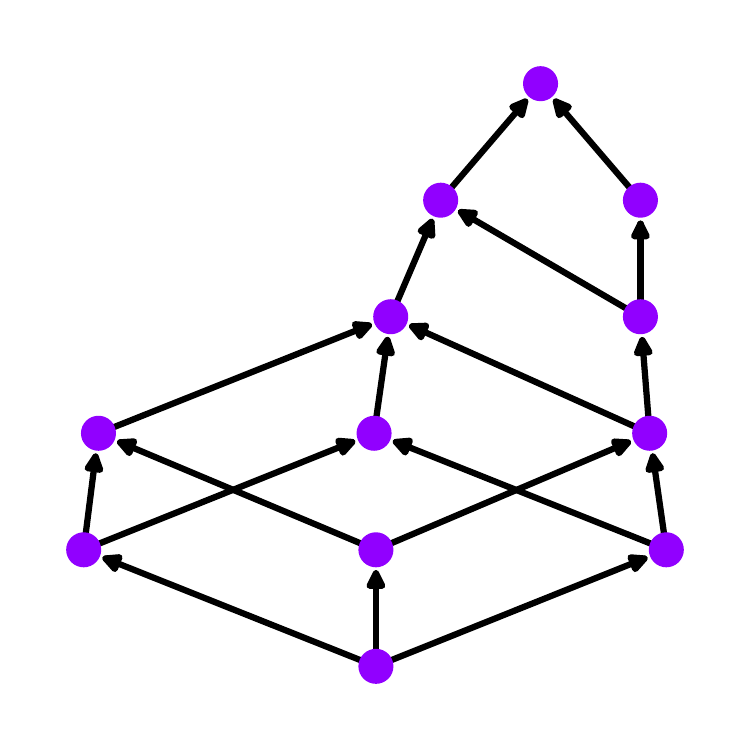} & 12 & 10 \\
 &  \includegraphics[height = 0.9cm, align = c]{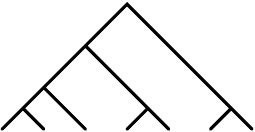} & \includegraphics[height = 1.6cm, align = c]{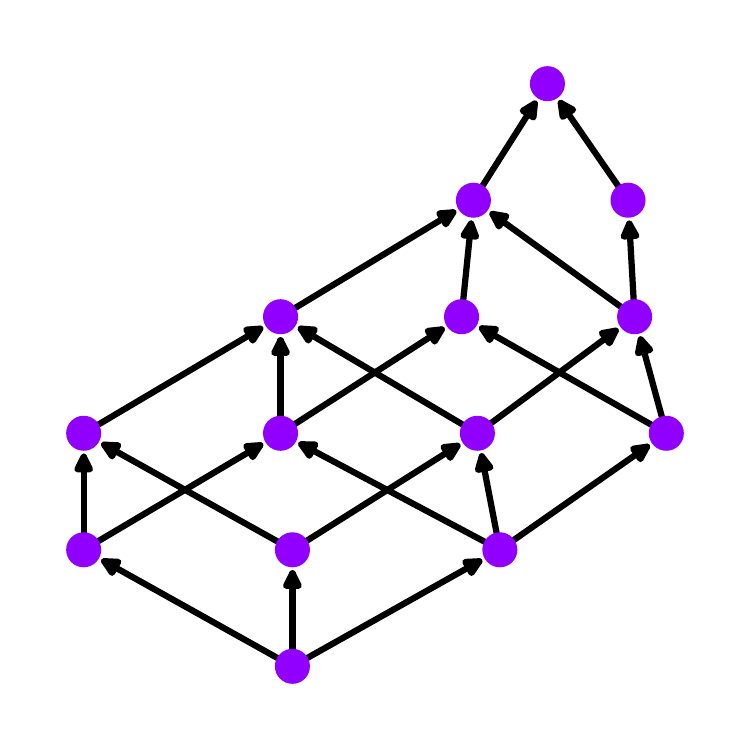} & 14 & 15 \\
 &  \includegraphics[height = 0.9cm, align = c]{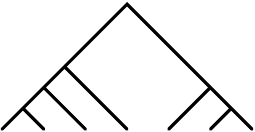} & \includegraphics[height = 1.6cm, align = c]{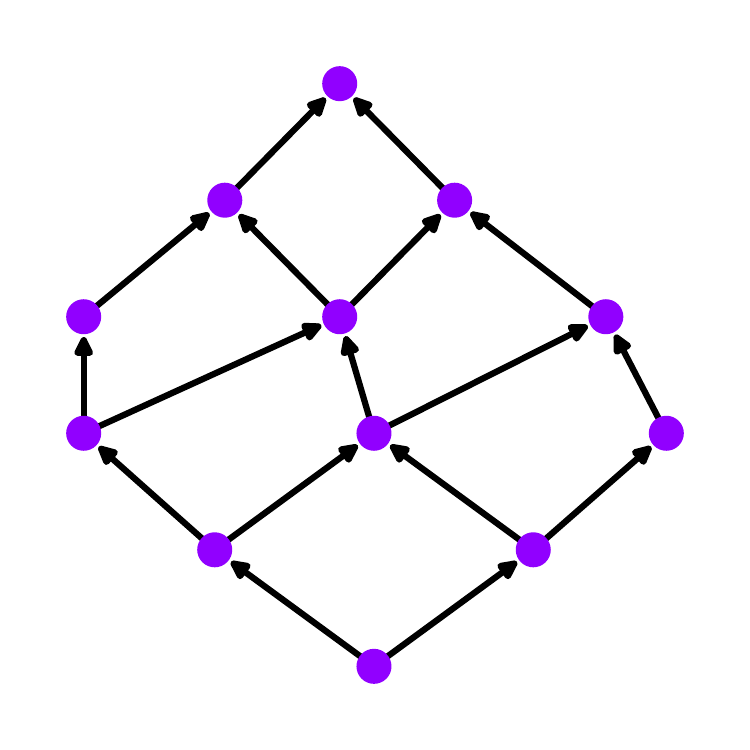} & 12 & 10 \\
 &  \includegraphics[height = 0.9cm, align = c]{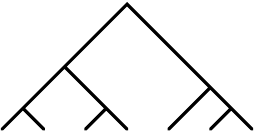} & \includegraphics[height = 1.6cm, align = c]{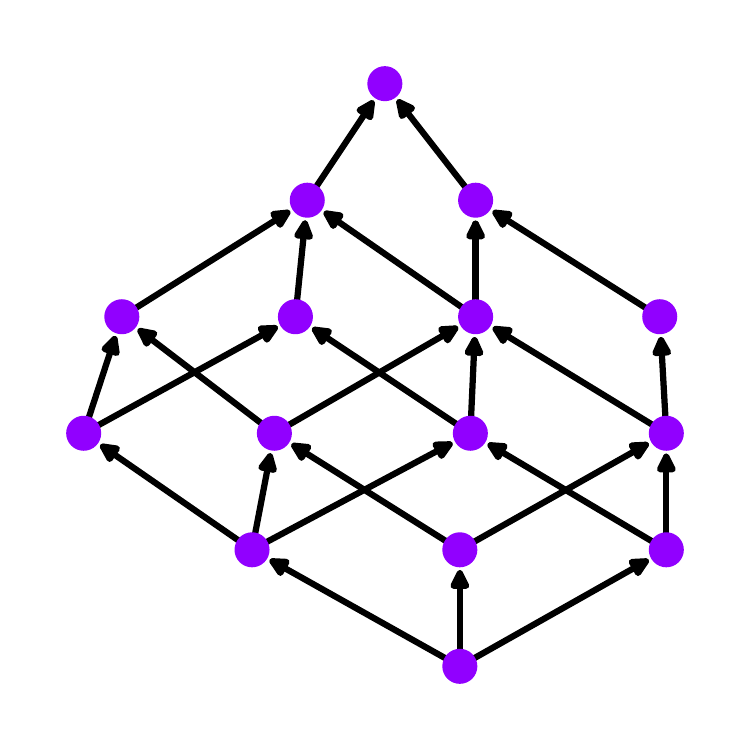} & 15 & 20 \\
\bottomrule
\end{tabular}
\end{minipage}}
\caption{Hasse diagrams $D(S)$ of the lattice of root ancestral configuations of $S$, for all trees $S$ of size $3 \leq n \leq 7$. A representative labeling for each unlabeled topology can be assumed, with the labels omitted for convenience. $c(S)$ is the number of root ancestral configurations of $S$, equal to the number of vertices in the Hasse diagram $D(S)$. $h(S)$ is the number of labeled histories of $S$, equal to the number of paths from the minimal element to the maximal element of $D(S)$.}
\label{tbl:examples_n_3_to_7}
\end{table}

\end{adjustwidth}

%\newpage

\begin{adjustwidth}{-30pt}{30pt}

\begin{table}
\adjustbox{valign=t}{\begin{minipage}{0.4\textwidth}
\begin{tabular}{ccccc}
\toprule
$n$ & Tree $S$ &$D(S)$ & $c(S)$ & $h(S)$ \\
\midrule
 
8 &  \includegraphics[height = 0.9cm, align = c]{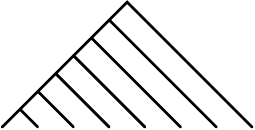} & \includegraphics[height = 1.6cm, align = c]{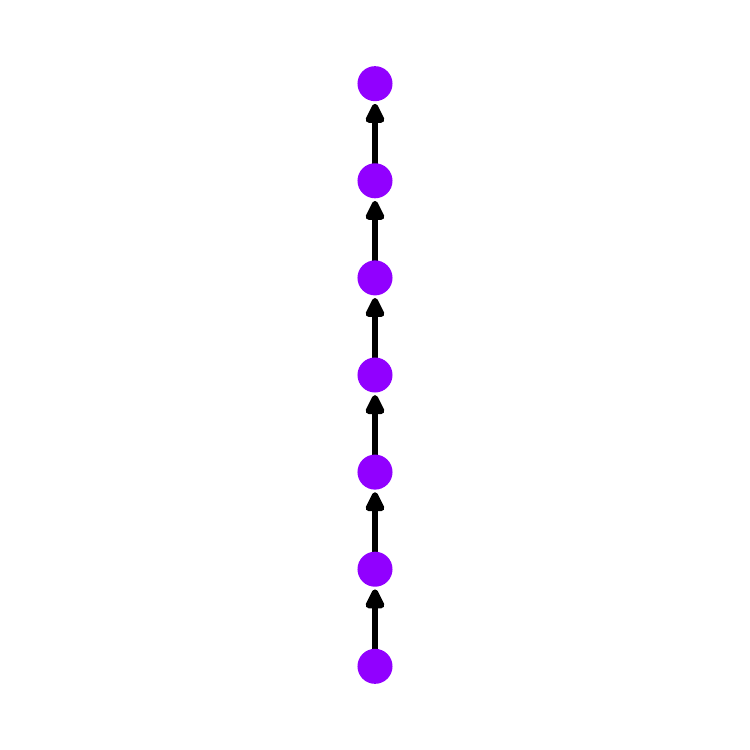} & 7 & 1 \\
 &  \includegraphics[height = 0.9cm, align = c]{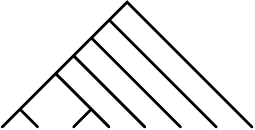} & \includegraphics[height = 1.6cm, align = c]{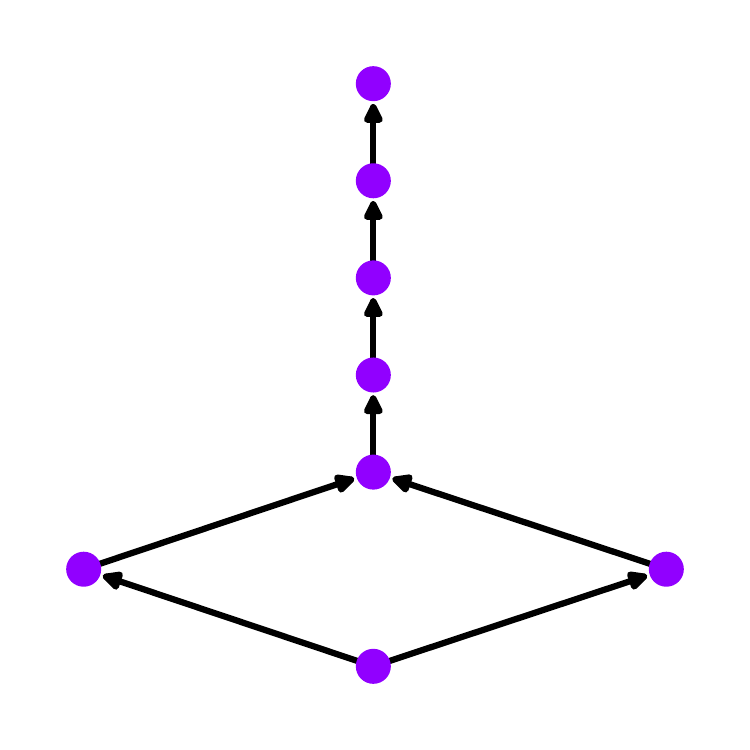} & 8 & 2 \\
 &  \includegraphics[height = 0.9cm, align = c]{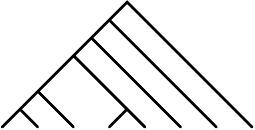} & \includegraphics[height = 1.6cm, align = c]{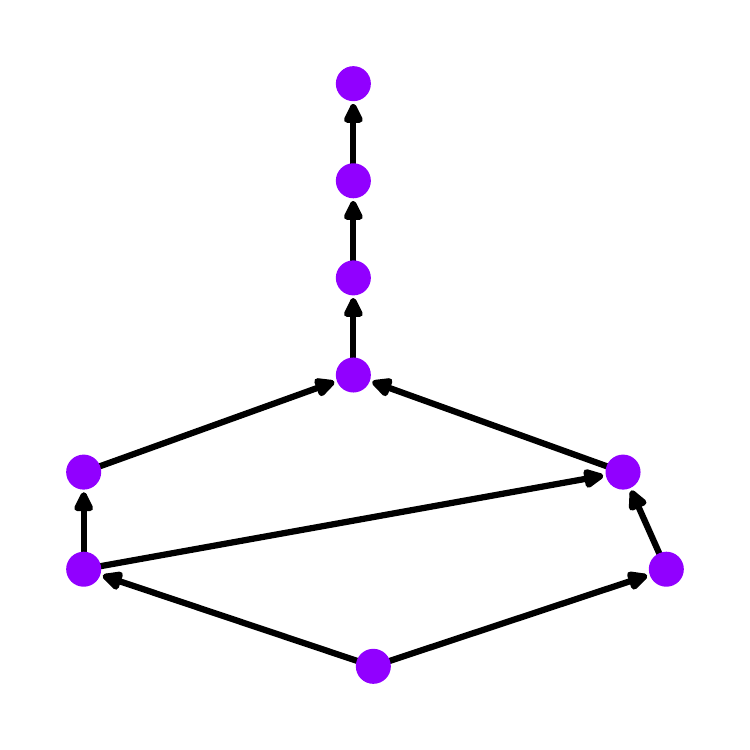} & 9 & 3 \\
 &  \includegraphics[height = 0.9cm, align = c]{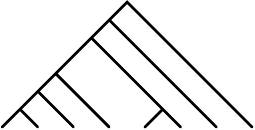} & \includegraphics[height = 1.6cm, align = c]{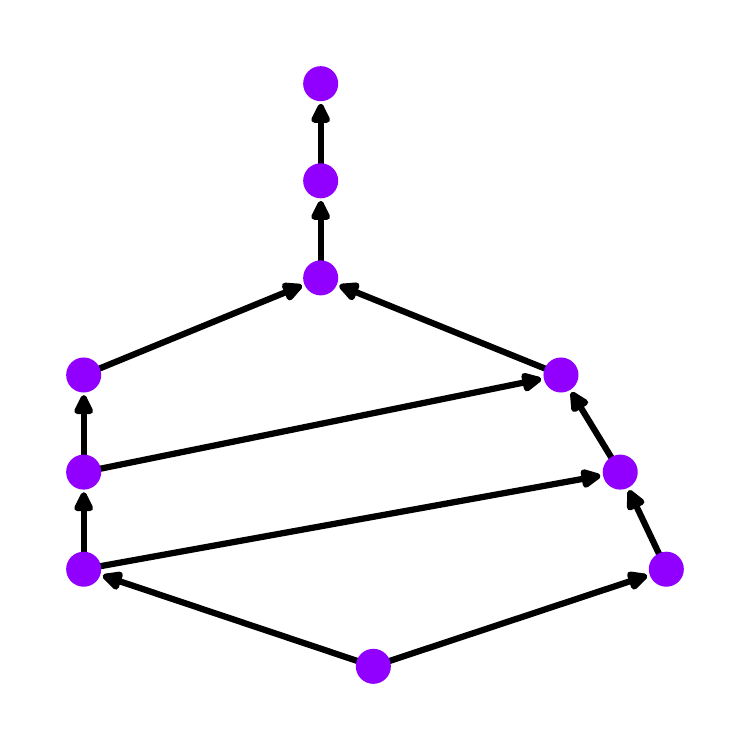} & 10 & 4 \\
 &  \includegraphics[height = 0.9cm, align = c]{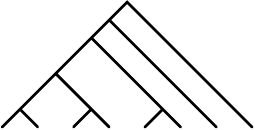} & \includegraphics[height = 1.6cm, align = c]{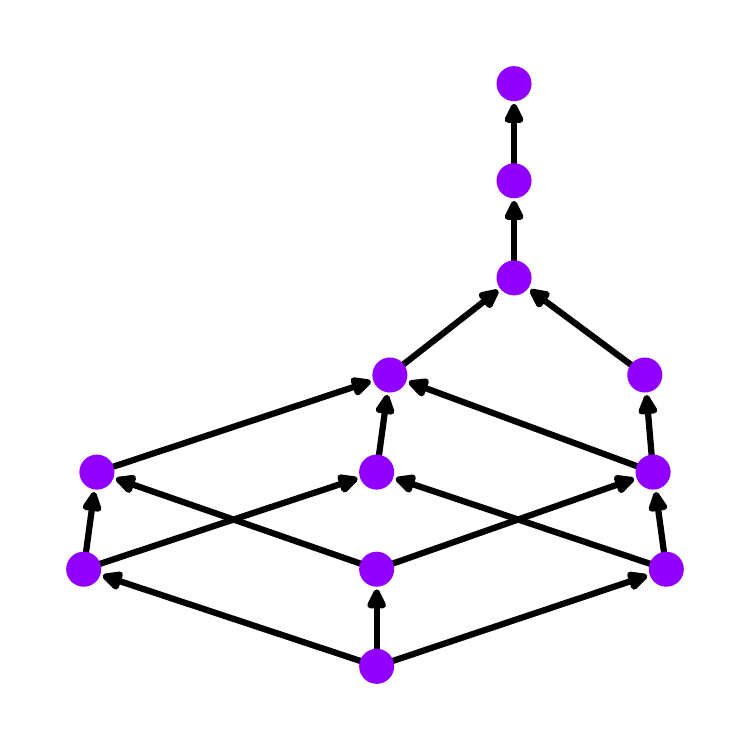} & 12 &  8 \\
 &  \includegraphics[height = 0.9cm, align = c]{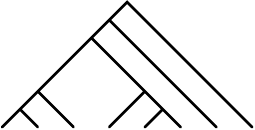} & \includegraphics[height = 1.6cm, align = c]{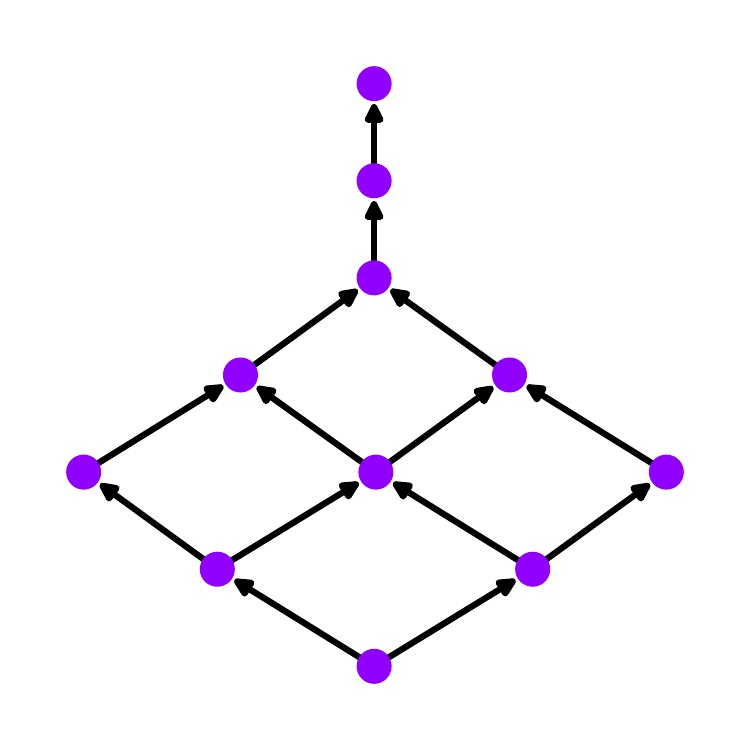} & 11 & 6 \\
 &  \includegraphics[height = 0.9cm, align = c]{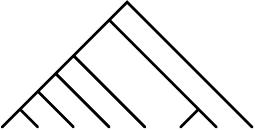} & \includegraphics[height = 1.6cm, align = c]{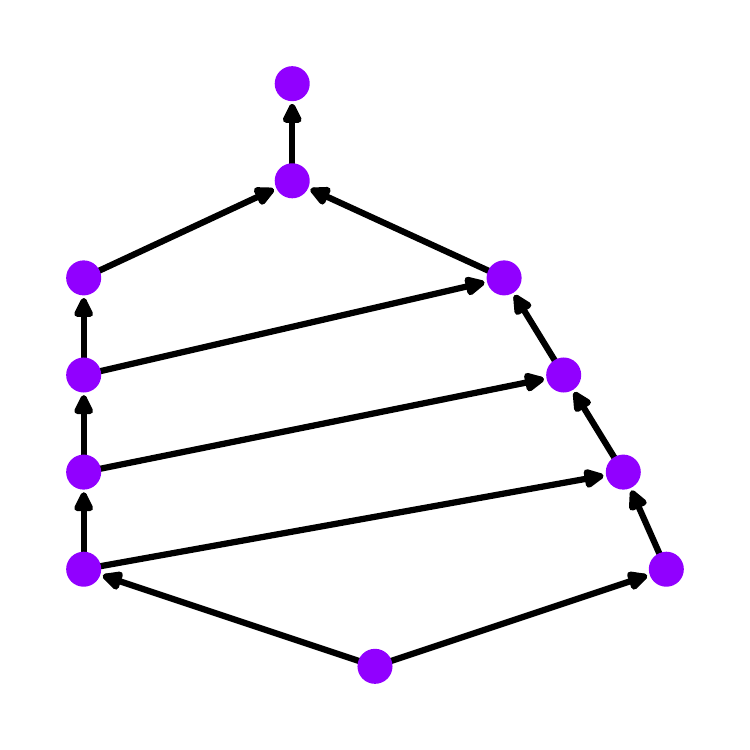} & 11 & 5 \\
 &  \includegraphics[height = 0.9cm, align = c]{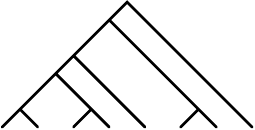} & \includegraphics[height = 1.6cm, align = c]{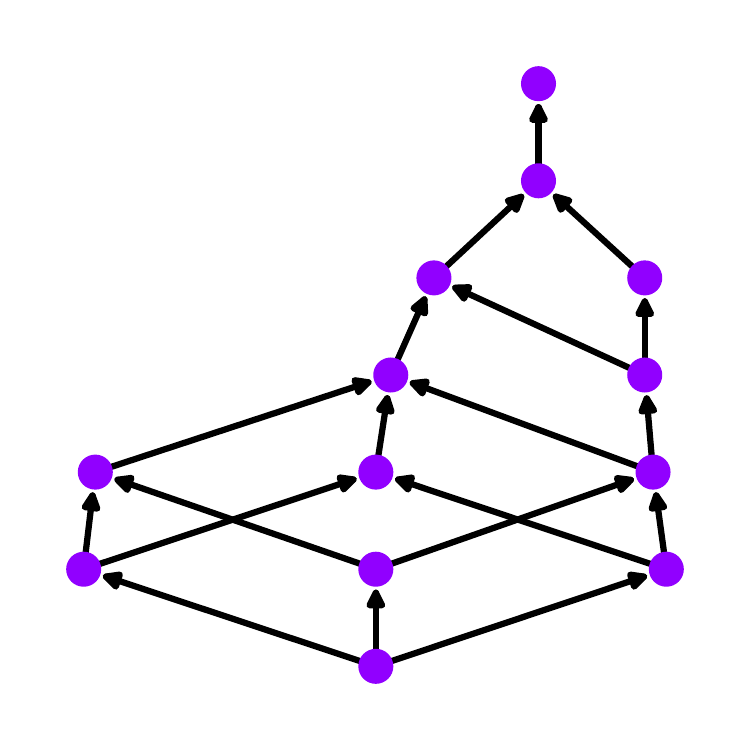} & 13 & 10 \\
 &  \includegraphics[height = 0.9cm, align = c]{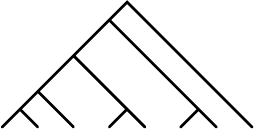} & \includegraphics[height = 1.6cm, align = c]{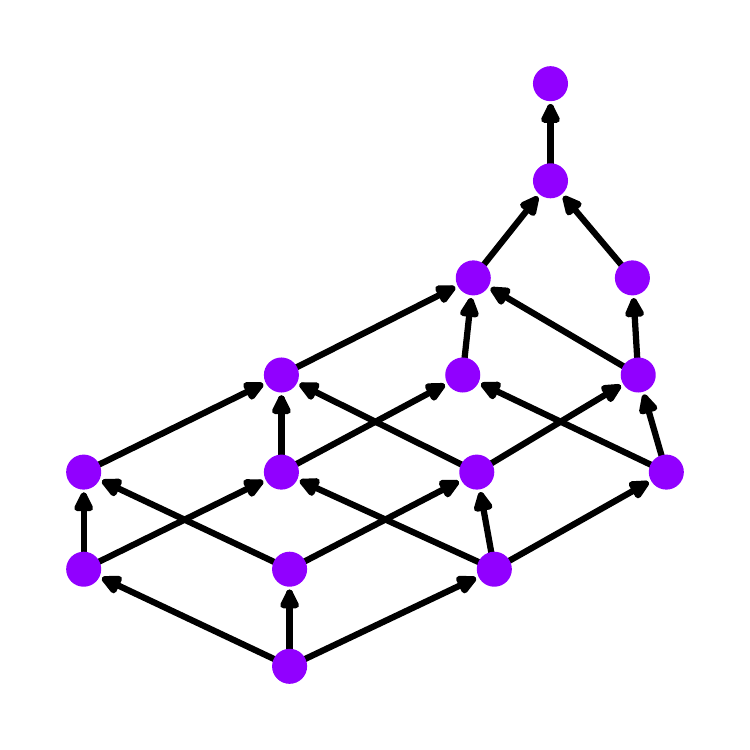} & 15 & 15 \\
 &  \includegraphics[height = 0.9cm, align = c]{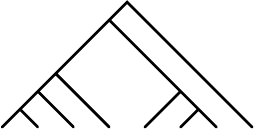} & \includegraphics[height = 1.6cm, align = c]{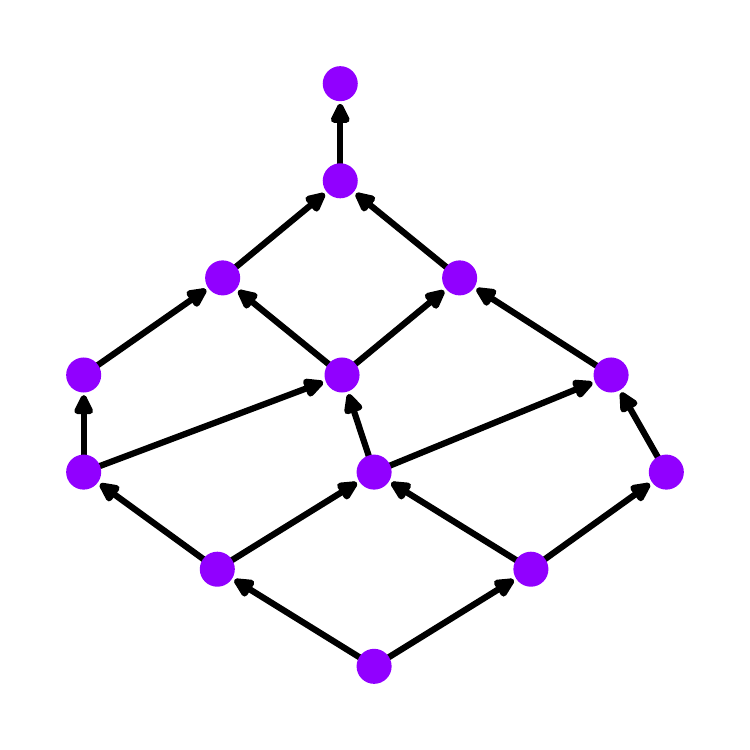} & 13 & 10 \\
 &  \includegraphics[height = 0.9cm, align = c]{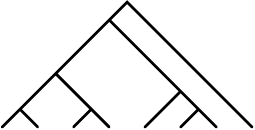} & \includegraphics[height = 1.6cm, align = c]{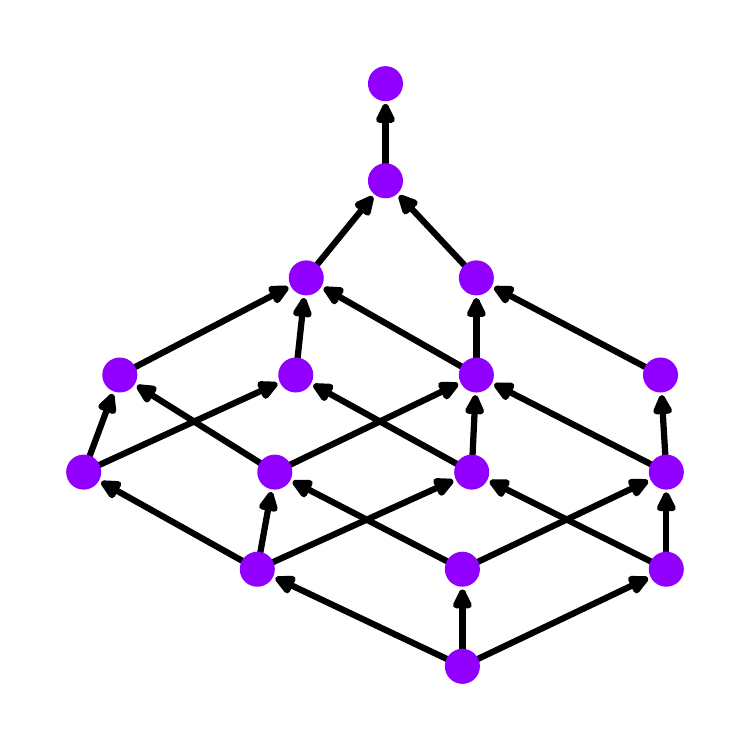} & 16 & 20 \\&  \includegraphics[height = 0.9cm, align = c]{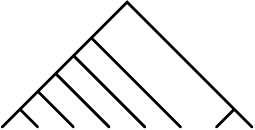} & \includegraphics[height = 1.6cm, align = c]{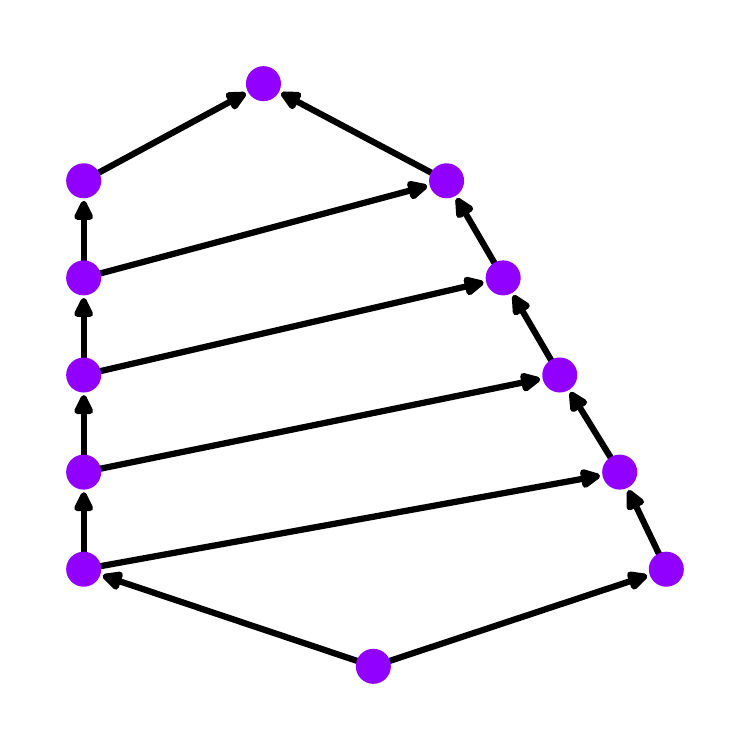} & 12 & 6 \\
\bottomrule
\end{tabular}
\end{minipage}} \hfill
\adjustbox{valign=t}{\begin{minipage}{0.4\textwidth}
\begin{tabular}{ccccc}
\toprule
$n$ & Tree $S$ &$D(S)$ & $c(S)$ & $h(S)$ \\
\midrule
 &  \includegraphics[height = 0.9cm, align = c]{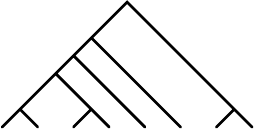} & \includegraphics[height = 1.6cm, align = c]{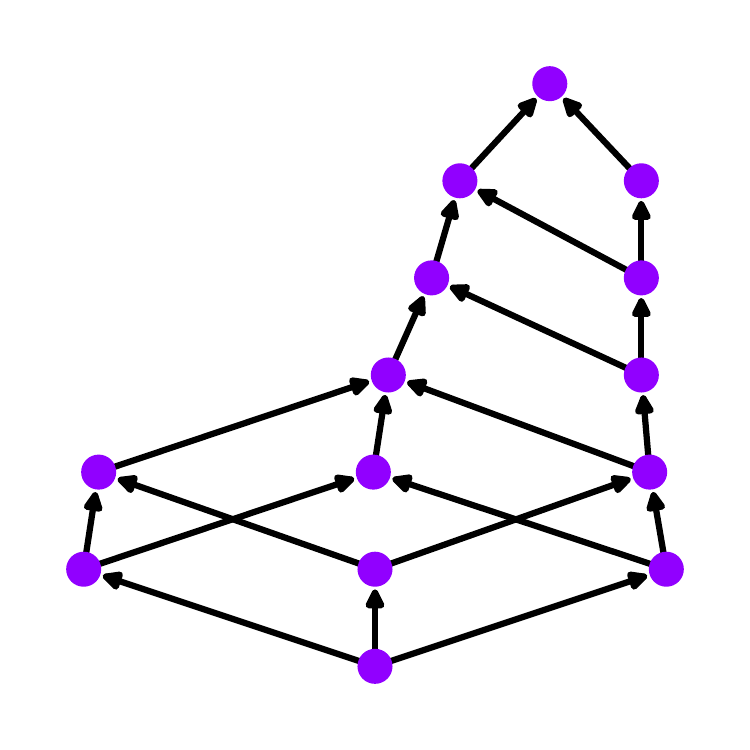} & 14 & 12 \\
 &  \includegraphics[height = 0.9cm, align = c]{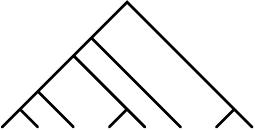} & \includegraphics[height = 1.6cm, align = c]{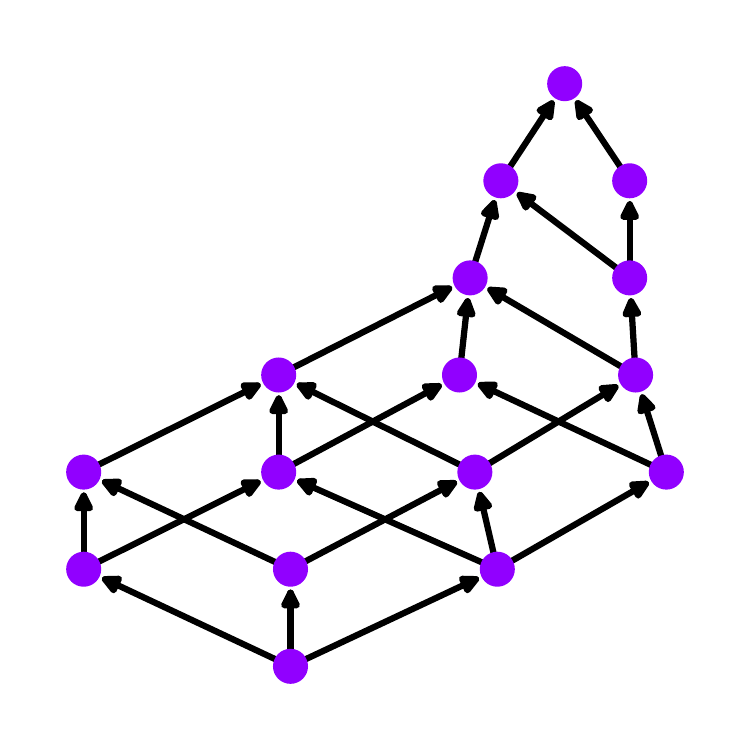} & 16 & 18 \\
 &  \includegraphics[height = 0.9cm, align = c]{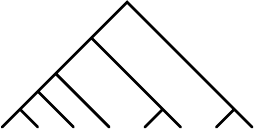} & \includegraphics[height = 1.6cm, align = c]{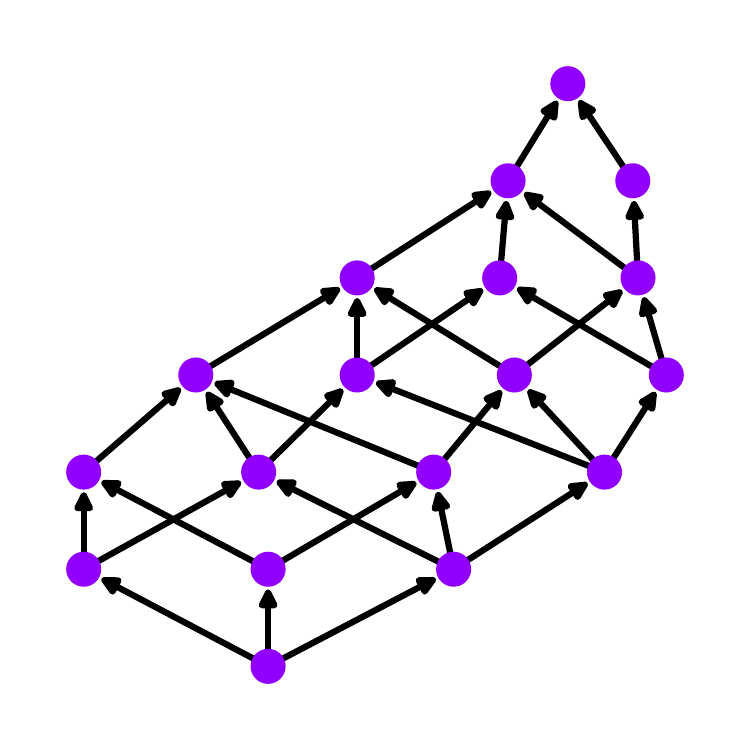} & 18 & 24 \\
  &  \includegraphics[height = 0.9cm, align = c]{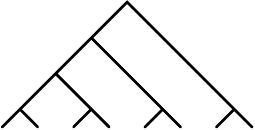} & \includegraphics[height = 1.6cm, align = c]{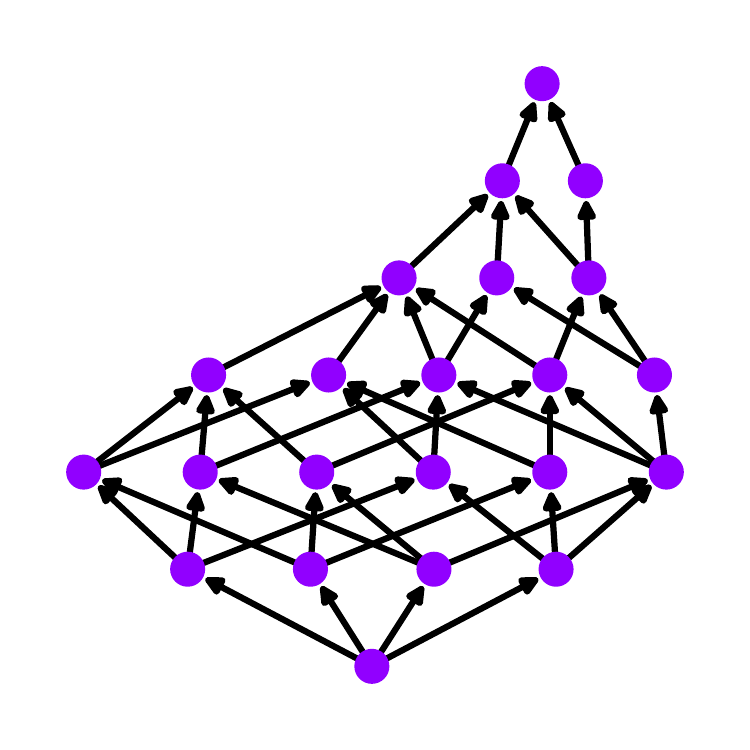} & 22 & 48 \\
 &  \includegraphics[height = 0.9cm, align = c]{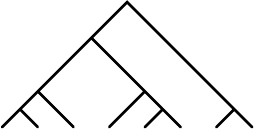} & \includegraphics[height = 1.6cm, align = c]{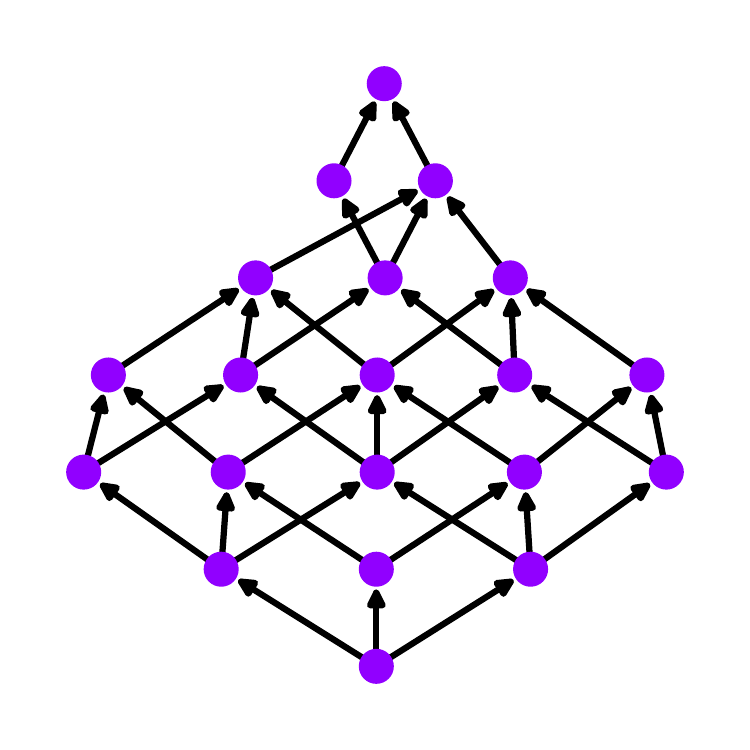} & 20 & 36 \\
 &  \includegraphics[height = 0.9cm, align = c]{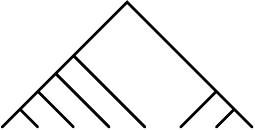} & \includegraphics[height = 1.6cm, align = c]{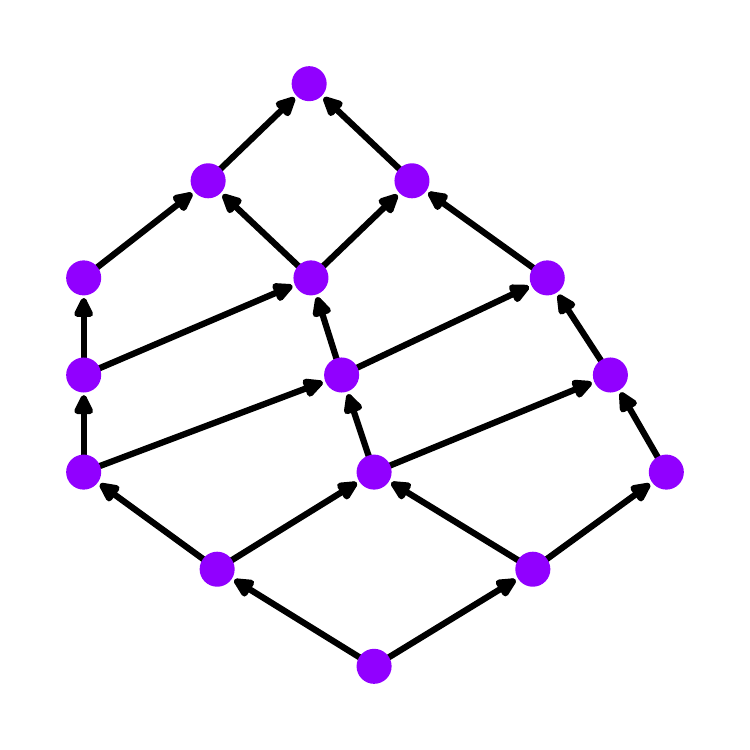} & 15 & 15 \\
 &  \includegraphics[height = 0.9cm, align = c]{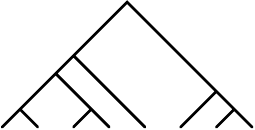} & \includegraphics[height = 1.6cm, align = c]{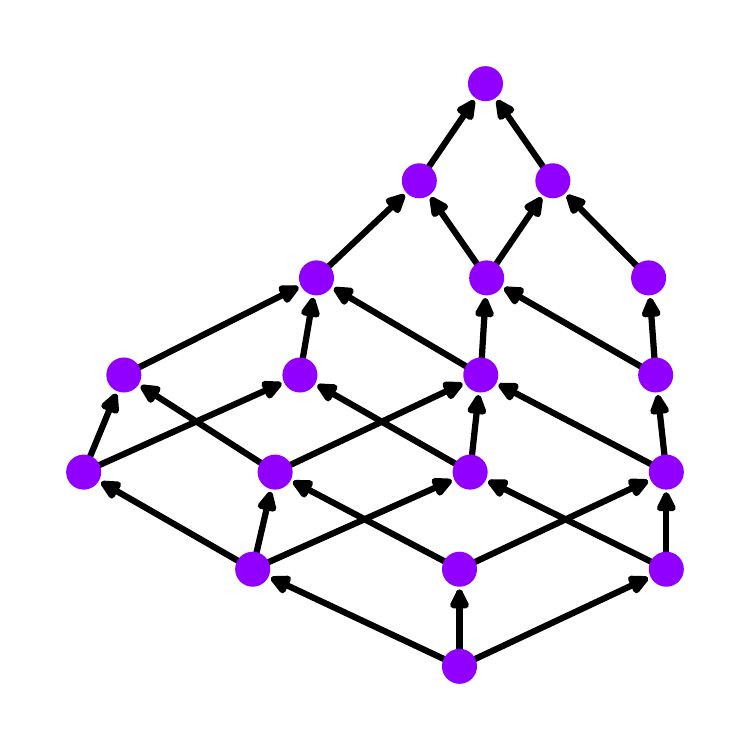} & 18 & 30 \\
 &  \includegraphics[height = 0.9cm, align = c]{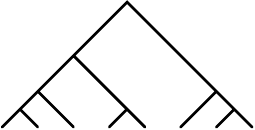} & \includegraphics[height = 1.6cm, align = c]{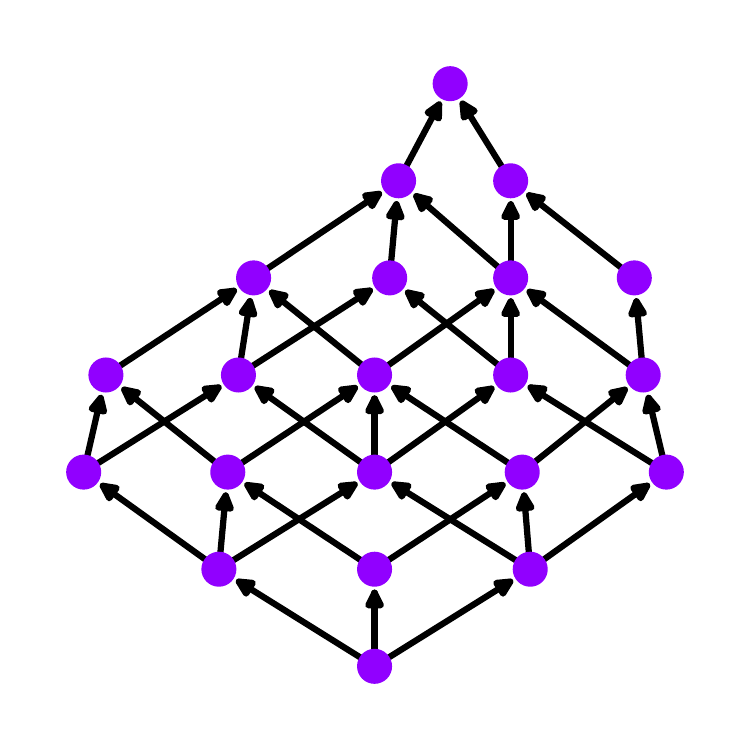} & 21 & 45 \\
 &  \includegraphics[height = 0.9cm, align = c]{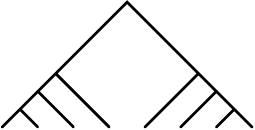} & \includegraphics[height = 1.6cm, align = c]{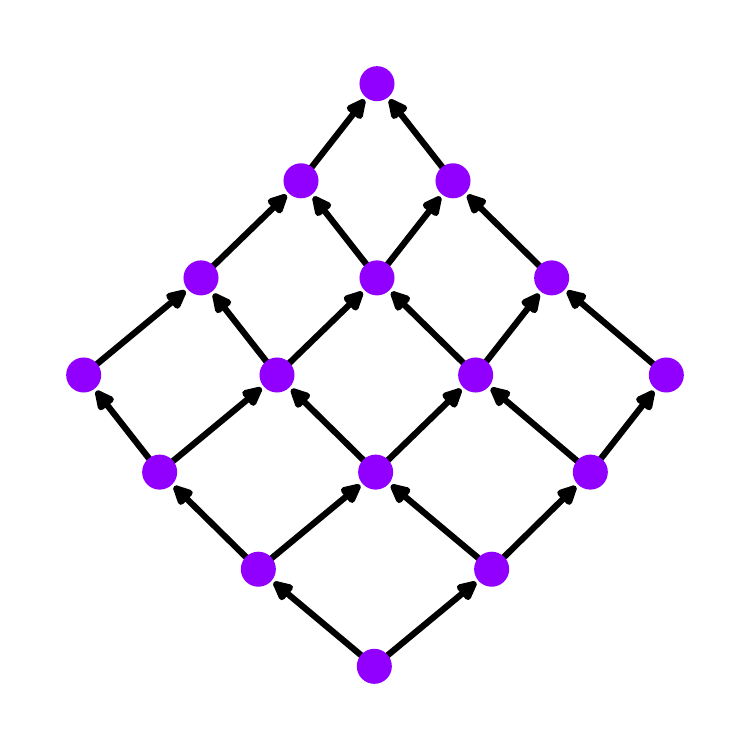} & 16 & 20 \\
 &  \includegraphics[height = 0.9cm, align = c]{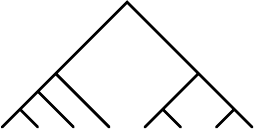} & \includegraphics[height = 1.6cm, align = c]{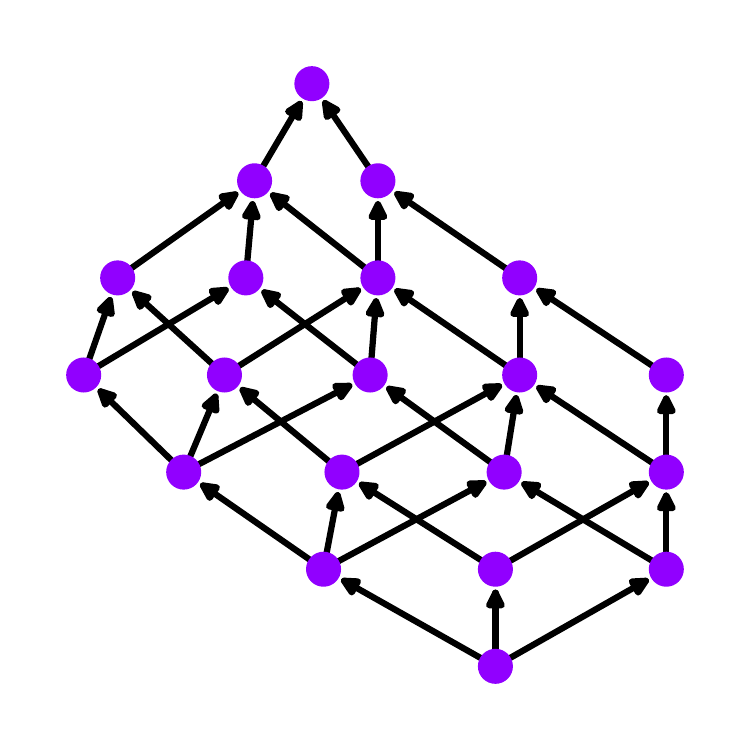} & 20 & 40 \\
 &  \includegraphics[height = 0.9cm, align = c]{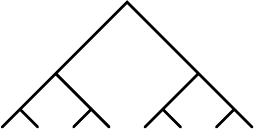} & \includegraphics[height = 1.6cm, align = c]{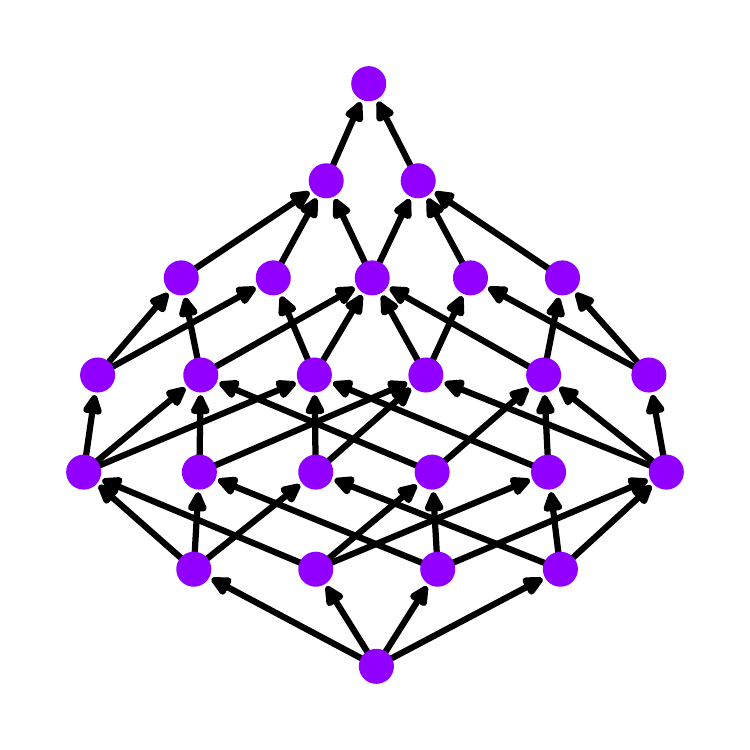} & 25 & 80 \\
\bottomrule
\end{tabular}
\end{minipage}}
\caption{Hasse diagrams $D(S)$ of the lattice of root ancestral configuations of $S$, for all trees $S$ of size $n=8$. As in Table \ref{tbl:examples_n_3_to_7}, a representative labeling for each unlabeled topology can be assumed, with the labels omitted for convenience. $c(S)$ is the number of root ancestral configurations of $S$, equal to the number of vertices in the Hasse diagram $D(S)$. $h(S)$ is the number of labeled histories of $S$, equal to the number of paths from the minimal element to the maximal element of $D(S)$.}
\label{tbl:examples_n_8}
\end{table}

\end{adjustwidth}

\clearpage
%%%%%%%%%%%%%%%%%%%%%%%%%%%%% Section 7 %%%%%%%%%%%%%%%%%%%%%%%%%%%%%%%%%
\section{A structure theorem}\label{sec:structure_description}

Based on the examples in Section~\ref{sec:examples}, we see that the Hasse diagrams $D(S)$ possess a recursive structure that mimics a recurrence for $C(S)$. Recall that a recurrence exists for root ancestral configurations of a given tree $S$ \citep[Proposition~1]{disanto2017enumeration}. First, for some sets of sets $A$ and $B$, define a binary (associative, commutative) operation $\otimes$ by
\[ A \otimes B = \{ a \cup b \mid a \in A, b\in B\}.\]
For a tree $S$ whose two nodes immediately descended from the root are labeled $r_\ell$ and $r_r$, the root ancestral configurations of $S$ can be written 
\begin{equation}
\label{eq:recur_enum}
	C(S) = \{ 
	\left[C(S|_{r_\ell}) \otimes C(S|_{r_r})\right] \cup 
	\left[\{\{r_\ell\}\} \otimes C(S|_{r_r})\right] \cup
    \left[C(S|_{r_\ell}) \otimes \{\{r_r\}\}\right] \cup 
	\{r_\ell,r_r\}\},
\end{equation}
where $S|_{r_\ell}$ and $S|_{r_r}$ represent the subtrees of $S$ rooted at $r_{\ell}$ and $r_r$, respectively. The cardinalities satisfy 
\begin{equation}
\label{eq:recur_count}
	c(S) = \left[c(S|_{r_\ell})+1\right] \left[c(S|_{r_r})+1\right].
\end{equation}

We prove an analogous recurrence for the construction of the lattices $D(S)$. First, a definition of the cartesian product of graphs \citep{sabidussi1959graph} is required.

\begin{definition}\label{def:cart_product}
	Given two directed graphs $G$ and $H$, their \emph{cartesian product} is a graph $\cp G H$ such that
	\begin{enumerate}
		\item The vertex set $V(\cp G H)$ is a cartesian product $V(G) \times V(H)$, and
		\item There is an edge from a vertex $(u,v)$ in $\cp G H$ to a vertex $(u',v')$ if either (i) $u=u'$ and there is an edge from $v$ to $v'$ in $H$ \emph{or} (ii) $v=v'$ and there is an edge from $u$ to $u'$ in $G$.
	\end{enumerate} 
\end{definition}

We also need some language for describing digraphs in the context  of ancestral configurations. Let $\mathcal{D}$ be the set of all directed graphs that arise as $D(S)$ for some labeled topology $S$. Recall that for all trees $S$, the graph $D(S)$ has a unique \emph{maximal node}: the node with no outgoing edges. This maximal node of $D(S)$ corresponds to the unique maximal root configuration in the lattice $C(S)$, consisting of two lineages immediately descended from the root of $S$. If one of these root lineages of $S$ is a leaf, then the maximal node of $D(S)$ has only one incoming edge. Hence, we have the following definition.
\begin{definition}
	We say that a digraph $D \in \mathcal D$ has a \emph{caterpillar-like maximal node} if the maximal node of $D$ has exactly one incoming edge.
\end{definition}

We introduce notation for appending caterpillar-like maximal nodes. For a digraph $D\in \mathcal D$, let $\widehat D$ be a digraph obtained by adding a new node $m'$ and an edge from the maximal node $m$ of $D$ to $m'$. Naturally, $m'$ is now a caterpillar-like maximal node of $\widehat D$. If $D$ is an empty graph, then we let $\widehat{D}$ be the graph with a unique node $m'$ and no edges.
For trees $S$, we interpret $\widehat{D(S)}$ as $D(S_1)$, where $S_1$ is the first tree in the caterpillar-like family with a seed tree $S$, as defined in Section~\ref{sec:chains}. With these definitions, the main result of the section can now be stated.

\begin{theorem}\label{thm:graph_decomp}
Consider a tree $S$ with $n\geq 2$ leaves. Let $r_\ell$ and $r_r$ be the two immediate descendants of its root, and let $S|_{r_\ell}$ and $S|_{r_r}$ be the subtrees of $S$ rooted at $r_\ell$ and $r_r$, respectively. Then the digraph $D(S)$ is isomorphic to a cartesian product of $\widehat{D(S|_{r_\ell})}$ and $\widehat{D(S|_{r_r})}$,
\[
	D(S) \cong \cp{\widehat{D(S|_{r_\ell})}}{\widehat{D(S|_{r_r})}}.
\]
\end{theorem}
Taking cardinalities of the vertex sets of the diagrams, we obtain eq.~\ref{eq:recur_count}, as $|\widehat{D(S)}| = |D(S)| + 1 = c(S) + 1$ for all $S$, so that $c(S) = |D(S)|= |\widehat{D(S|_{r_\ell})}| |\widehat{D(S|_{r_r})}| = \left[c(S|_{r_\ell})+1\right] \left[c(S|_{r_r})+1\right]$. Theorem~\ref{thm:graph_decomp} thus expands the recurrences in eqs.~\ref{eq:recur_enum} and \ref{eq:recur_count} to the level of Hasse diagrams. The diagram for a tree $S$ is a cartesian product of the diagrams for its root subtrees $S|_{r_r}$ and $S|_{r_\ell}$ with an extra node added to each, to account for root coalescences of $S|_{r_r}$ and $S|_{r_\ell}$.

\begin{proof}[Proof of Theorem~\ref{thm:graph_decomp}]
Eq.~\ref{eq:recur_enum} gives a bijection between the vertex sets $V\big(D(S)\big)$ and $V\left(\cp{\widehat{D(S|_{r_\ell})}}{\widehat{D(S|_{r_r})}}\right)$ by interpreting caterpillar-like maximal nodes added by the ``hat''-operation as $r_r$ and $r_\ell$ in eq.~\ref{eq:recur_enum}. This bijection allows us to formally identify the vertices of the two digraphs.

To conclude that $D(S) = \cp{\widehat{D(S|_{r_\ell})}}{\widehat{D(S|_{r_r})}}$, we must show that the sets of edges of digraphs are equal. We first show containment of edge sets, and then prove by contradiction that containment cannot be strict.

We claim that edge sets satisfy $E\left(\cp{\widehat{D(S|_{r_\ell})}}{\widehat{D(S|_{r_r})}}\right) \subseteq E\big(D(S)\big)$ under the identification of vertices via eq.~\ref{eq:recur_enum}. The digraph $\cp{\widehat{D(S|_{r_\ell})}}{\widehat{D(S|_{r_r})}}$ has an edge between $(u,v)$ and $(u',v')$ if either (1) $u=u'$ and $\widehat{D(S|_{r_r})}$ has an edge from $v$ to $v'$ or (2) if $v = v'$ and $\widehat{D(S|_{r_\ell})}$ has an edge from $u$ to $u'$. This means that each edge of $\cp{\widehat{D(S|_{r_\ell})}}{\widehat{D(S|_{r_r})}}$, viewed as an edge between vertices of $D(S)$, corresponds to precisely one coalescence, happening either in the right subtree $S_r$ (case 1), or in the left subtree $S_\ell$ (case 2), including the root coalescences of the two subtrees. Such an edge is in $E\big(D(S) \big)$ by definition of $D(S)$.

Finally, if there is an edge $e \in E\big(D(S)\big) - E\left(\cp{\widehat{D(S|_{r_\ell})}}{\widehat{D(S|_{r_r})}}\right)$, then it corresponds to a coalescence that happens in neither of the root subtrees $S_\ell$, $S_r$. There is only one such coalescence, namely the root of $S$. However, by definition of $D(S)$, root coalescences are not reflected in the digraph, so we have a contradiction. We conclude that $D(S) = \cp{\widehat{D(S|_{r_\ell})}}{\widehat{D(S|_{r_r})}}$.
\end{proof}

To simplify notation, we define a binary operation on digraphs $D_1$ and $D_2$, $\widehat\square$, by \[\cphat{D_1}{D_2} = \cp{\widehat{D_1}}{\widehat{D_2}}.\]
This operation combines the operations of adding ancestral nodes and taking cartesian products. Hence, Theorem~\ref{thm:graph_decomp} allows us to encode a graph as a $\widehat\square$-product of simpler graphs. 

First, recall that for a tree consisting of only one leaf, we define the set of ancestral configurations, $C(S)$, to be empty. Correspondingly, $D(S)$ is an empty graph. For a tree with two leaves, $S = (a,b)$, the set of ancestral configurations $C(S)$ has only one element, and $D(S)$ is a graph with one node and no edges. 

If we recursively decompose a tree $S$ into left and right subtrees $S|_{r_\ell}$ and $S|_{r_r}$, then we eventually arrive at trees with one leaf and trees with two leaves (cherry nodes). Of course, a tree with two leaves decomposes into two one-leaf trees, but it is convenient not to make this decomposition.

Repeated application of Theorem~\ref{thm:graph_decomp} gives us the following result.
\begin{proposition}\label{prop:D_desc}
	Let $\emptyset$ denote an empty graph, and let $I$ denote the graph with a single node and no edges. Then any digraph $D \in \mathcal D$ can be obtained by repeated application of the $\widehat\square$ operator to $\emptyset$ and $I$.
\end{proposition}
The decomposition can be ``read'' directly from a tree by writing $I$ for each cherry node and $\emptyset$ for each single edge, including $\widehat\square$ products accordingly. For example, for a caterpillar tree $\mathcal C_n$ with $n\geq 2$, we write 
\begin{equation}\label{eq:decomp_cat}
    D(\mathcal C_n) = (((I \house \emptyset) \ldots) \house \emptyset) \house \emptyset,
\end{equation}
with the $\widehat \square$-product applied $n-2$ times. As $\widehat{\emptyset} = I$, $D(\mathcal C_n)$ is equal to $I$ with $n-2$ caterpillar-like nodes added, giving us a linear path graph, as seen in Section~\ref{sec:examples}.

In the next section, we present a series of examples of the application of Theorem \ref{thm:graph_decomp} and Proposition \ref{prop:D_desc}.

%%%%%%%%%%%%%%%%%%%%%%%%%% Setion 8 %%%%%%%%%%%%%%%%%%%%%%%%%%%%%%%%%%%%%%
\section{Specific families}\label{sec:families}

With the connection between sets of ancestral configurations and lattices established, we now examine the diagrams produced by special families of trees. As noted in Section \ref{sec:chains}, the diagrams for the caterpillar trees are trivial path graphs, with a unique maximal chain. Here, we consider $p$-pseudocaterpillar trees, bicaterpillar trees, two types of lodgepole trees, and balanced trees. In the figures of this section, we omit leaf and node labels on trees and diagrams for visual clarity.

%%%%%%%%%%%%%%%%%%%%%%%%%% Section 8.1 %%%%%%%%%%%%%%%%%%%%%%%%%%%%%%%%%%%%%
\subsection{Pseudocaterpillar trees}\label{sec:pcat_ex}

First, we consider $p$-pseudocaterpillar trees, as defined by \cite{alimpiev2021enumeration}. For $p \geq 4$, a $p$-pseudocaterpillar tree has two cherries; it is encoded by integers $(n,p)$, where $n$ is the number of leaves, and $p$ is an index that denotes the location of one of the cherries. More precisely, we label the leaves from left to right by labels 1 to $n$; one cherry contains leaves 1 and 2, and the other contains leaves $p-1$ and $p$. For $p=3$, a $p$-pseudocaterpillar tree has only one cherry, consisting of leaves labeled 2 and 3. The $p$-pseudocaterpillar trees encoded by $(n,n)$ are $(a,(b,c))$ for $n=3$, $((a,b),(c,d))$ for $n=4$, and $(((a,b),c),(d,e))$ for $n=5$. We denote the $p$-pseudocaterpillar tree encoded by $(n,n)$ with the symbol $\mathcal P_n$ (Figure~\ref{fig:families2}).

Hasse diagrams for the sets of root configurations for the pseudocaterpillar family $\mathcal P_n$ appear in Figure~\ref{fig:pcat}. Recall that the number of nodes in a Hasse diagram is the number of root configurations; the number of paths in a Hasse diagram is the number of labeled histories.

The diagrams have a structure in which a path from bottom to top on the left side of the diagram represents coalescence of the caterpillar branches prior to coalescence of the cherry containing leaves $n-1$ and $n$, and a path from bottom to top on the right side of the diagram represents coalescence of the caterpillar branches after coalescence of the cherry. The edges connecting the path on the left and the path on the right represent coalescence of the cherry. Using Theorem~\ref{thm:graph_decomp}, this pattern can be summarized by 
\begin{equation}\label{eq:decomp_pcat}
	D(\mathcal P_n) = \cphat{D(\mathcal C_{n-2})}{I}.
\end{equation}

The number of root ancestral configurations $c(\mathcal{P}_n)$, or the number of nodes in the Hasse diagram of $\mathcal{P}_n$, is computed by eq.~\ref{eq:recur_count}, and we obtain $c(\mathcal{P}_n) = 2(n-2)$. 
Similarly, the number of labeled histories $h(\mathcal{P}_n)$, or the number of paths in the Hasse diagram of $\mathcal{P}_n$, is calculated by eq.~\ref{eq:steel_exp}, equaling $n-2$. The values of $c(\mathcal{P}_n)$ and $h(\mathcal{P}_n)$ are given in Table~\ref{tbl:pcat_numbers}.

%%%%%%%%%%%%%%%%%%% Table 8 %%%%%%%%%%%%%%%%%%%%%%%%%%%
\begin{table}[tb]
\centering
	\begin{tabular}{ccccccc}
		\toprule 
            & \multicolumn{6}{c}{$n$} \\ \cline{2-7}
		    & 3 & 4 & 5 & 6 & 7 & 8 \\
		\hline
		$c(\mathcal P_n)$ & 2 & 4 & 6 & 8 & 10 & 12 \\
		$h(\mathcal P_n)$ & 1 & 2 & 3 & 4 & 5  & 6  \\
		\bottomrule
	\end{tabular}
	\caption{The number of root ancestral configurations $c(\mathcal P_n)$ and the number of paths $h(\mathcal P_n)$ for $p$-pseudocaterpillar trees $\mathcal P_n$.}
	\label{tbl:pcat_numbers}
\end{table}
%%%%%%%%%%%%%%%%%%%%%%%%%%%%%%%%%%%%%%%%%%%%%%%%%%%%%%%

%%%%%%%%%%%%%%%%%%%%%%%%%%%%%% Figure 5 %%%%%%%%%%%%%%%%%%%%%%%%%%%%%%%%%
\begin{figure}[ht]\centering
\includegraphics[width=0.65\linewidth]{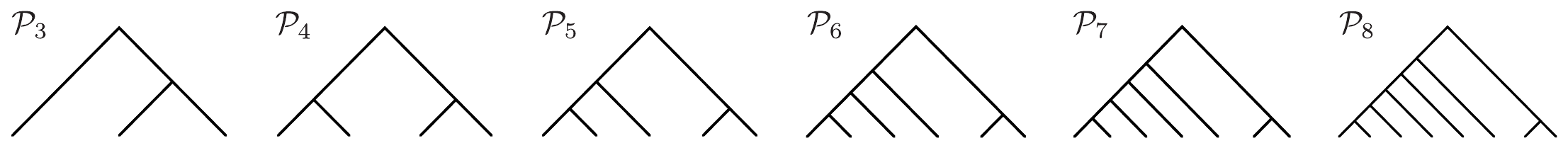}
\caption{$p$-pseudocaterpillar trees $\mathcal P_n$ for $n=3,4,5,6,7,8$.}
\label{fig:families2}
\end{figure}
%%%%%%%%%%%%%%%%%%%%%%%%%%%%%%%%%%%%%%%%%%%%%%%%%%%%%%%%%%%%%%%%%%%%%%%%%

%%%%%%%%%%%%%%%%%%%%%%%%%%%%%%% Figure 6 %%%%%%%%%%%%%%%%%%%%%%%%%%%%%%%%
\begin{figure}[ht]\centering
\includegraphics[width=0.65\linewidth, trim= 200 100 200 100, clip]{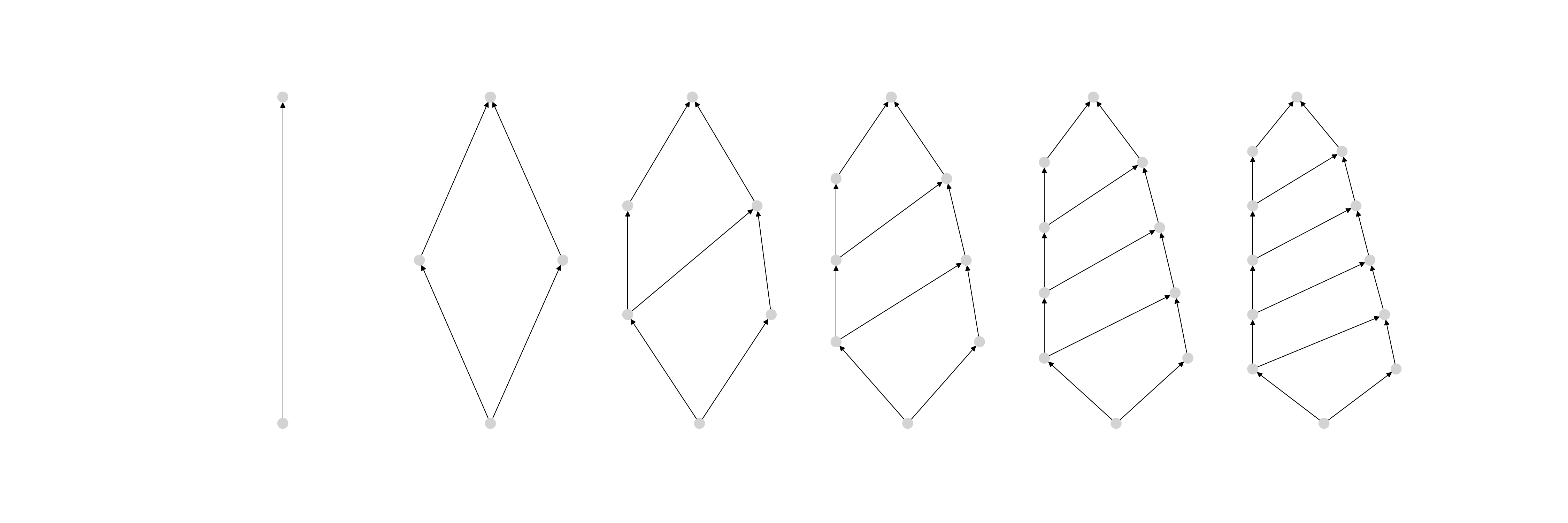}
\caption{Hasse diagrams corresponding to $p$-pseudocaterpillar trees $\mathcal P_n$ for $n = 3, 4, 5, 6, 7, 8$.}
\label{fig:pcat}
\end{figure}
%%%%%%%%%%%%%%%%%%%%%%%%%%%%%%%%%%%%%%%%%%%%%%%%%%%%%%%%%%%%%%%%%%%%%%%%%

%%%%%%%%%%%%%%%%%%%%%%%%%% Section 8.2 %%%%%%%%%%%%%%%%%%%%%%%%%%%%%%%%%%
\subsection{Bicaterpillars}
\label{subsec:bicat_ex}

For our next family, we consider symmetric bicaterpillar trees, defined as trees whose two root subtrees---trees induced by immediate descendants of the root---are caterpillar trees $\mathcal C_n$. We denote by $\mathcal{C}_{n,n}$ the symmetric bicaterpillar tree with two caterpillar subtrees $\mathcal{C}_n$ immediately descended from the root. In particular, $\mathcal C_{1,1} = (a,b)$, $\mathcal C_{2,2} = ((a,b),(c,d))$, $\mathcal C_{3,3} = (((a,b),c),(d,(e,f)))$, and so on, as shown in Figure~\ref{fig:families3}. 

%%%%%%%%%%%%%%%%%%% Table 4 %%%%%%%%%%%%%%%%%%%%%%%%%%%%%%%%%%%%%%%%%%%%%
\begin{table}[tb]
\centering
	\begin{tabular}{ccccccc}
		\toprule 
            & \multicolumn{6}{c}{$n$} \\ \cline{2-7}
		    & 1 & 2 & 3 & 4 & 5 & 6 \\
		\hline
		$c(\mathcal C_{n,n})$ & 1 & 4 & 9 & 16 & 25 & 36 \\
		$h(\mathcal C_{n,n})$ & 1 & 2 & 6 & 20 & 70 & 252 \\
		\bottomrule
	\end{tabular}
	\caption{The number of root ancestral configurations $c(\mathcal C_{n,n})$ and the number of paths $h(\mathcal C_{n,n})$ for bicaterpillar trees $\mathcal C_{n,n}$.}
	\label{tbl:bicat_numbers}
\end{table}
%%%%%%%%%%%%%%%%%%%%%%%%%%%%%%%%%%%%%%%%%%%%%%%%%%%%%%%

The Hasse diagrams for the sets of root configurations are presented in Figure~\ref{fig:bicat}. In the Hasse diagrams, edges moving up from left to right represent coalescences on say, the left side of the tree, and edges moving up from right to left represent coalescences on the right side of the tree.

The symmetry of the symmetric bicaterpillars generates a symmetry in the edges of the Hasse diagram, so that the Hasse diagram can be viewed as a square grid. The equivalence with paths from one corner of a square grid to an opposite corner is convenient for counting the paths. In particular, the number of nodes in the Hasse diagram, obtained by eq.~\ref{eq:recur_count}, is $c(\mathcal{C}_{n,n}) = n^2$. The number of paths, obtained by eq.~\ref{eq:steel_exp}, is $h(\mathcal{C}_{n,n})={2n-2 \choose n-1}$. The values of $c(\mathcal{C}_{n,n})$ and $h(\mathcal{C}_{n,n})$ appear in Table~\ref{tbl:bicat_numbers}.

Note that this example naturally extends to the case of general bicaterpillar trees $\mathcal C_{p,q}$, in which the two immediate subtrees of the root are caterpillars of size $p$ and $q$ leaves. The diagram is a rectangle with $pq$ nodes, and the number of paths on it is $\binom{p+q-2}{q-1}$.

Theorem~\ref{thm:graph_decomp} gives a quick proof of this result. By definition of $\mathcal C_{p,q}$, we have 
\begin{equation}\label{eq:decomp_bicat}
    D(\mathcal C_{p,q}) = \cphat{\mathcal C_p}{\mathcal C_q},
\end{equation}
so that $D(\mathcal C_{p,q})$ is a cartesian product of two path graphs with $p$ and $q$ vertices, respectively, or a ``$p \times q$ grid.'' 

The bicaterpillar example can be extended still further, to trees for which one of the two immediate subtrees of the root is a caterpillar.

%The stars and bars method used to show that $h(\mathcal C_{p,q}) = \binom{p+q-2}{q-1}$ 

\begin{proposition}\label{prop:house_with_caterpillar}
	Consider a tree whose two subtrees immediately descended from the root are $S$ and 
	$\mathcal{C}_p$, with $|S| \geq 1$ and $\mathcal{C}_p$ a caterpillar with $p$ leaves. Then $h(\cphat{D(S)}{D(\mathcal C_p)}) = \binom{n+p-1}{p-1} \,h\big(D(S)\big)$, where $n=|S|-1$, the number of coalescences of $S$.
\end{proposition}
\begin{proof}
    Given a path in $h\big(D(S)\big)$ (of length $n-1$), we can construct $\binom{n+p-1}{p-1}$ distinct paths in $h(\cphat{D(S)}{D(\mathcal C_p)})$, each with length $n+p-1$, by choosing the location of $p-1$ (linearly ordered) edges of $D(\mathcal C_p)$ among the $n+p-1$ total edges. Of course, paths constructed from distinct elements of $h\big(D(S)\big)$ are distinct as well, which allows us to conclude that $h(\cphat{D(S)}{D(\mathcal C_p)}) = \binom{n+p-1}{p-1} \,h\big(D(S)\big)$.
\end{proof}

%%%%%%%%%%%%%%%%%%%%%%%%%%%%%%% Figure 7 %%%%%%%%%%%%%%%%%%%%%%%%%%%%%%%%%%
\begin{figure}[ht]\centering
\includegraphics[width=0.65\linewidth]{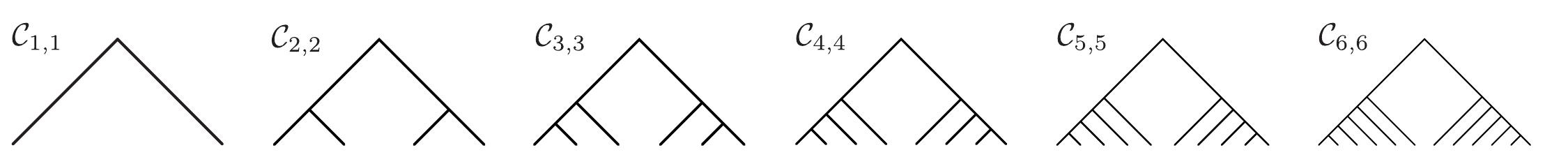}
\caption{Bicaterpillar trees $\mathcal C_{n,n}$ for $n=1,2,3,4,5,6$.}
\label{fig:families3}
\end{figure}
%%%%%%%%%%%%%%%%%%%%%%%%%%%%%%%%%%%%%%%%%%%%%%%%%%%%%%%%%%%%%%%%%%%%%%%%%%%

%%%%%%%%%%%%%%%%%%%%%%%%%%%%%%% Figure 8 %%%%%%%%%%%%%%%%%%%%%%%%%%%%%%%%%%
\begin{figure}[ht]\centering
\includegraphics[width=0.5\linewidth]{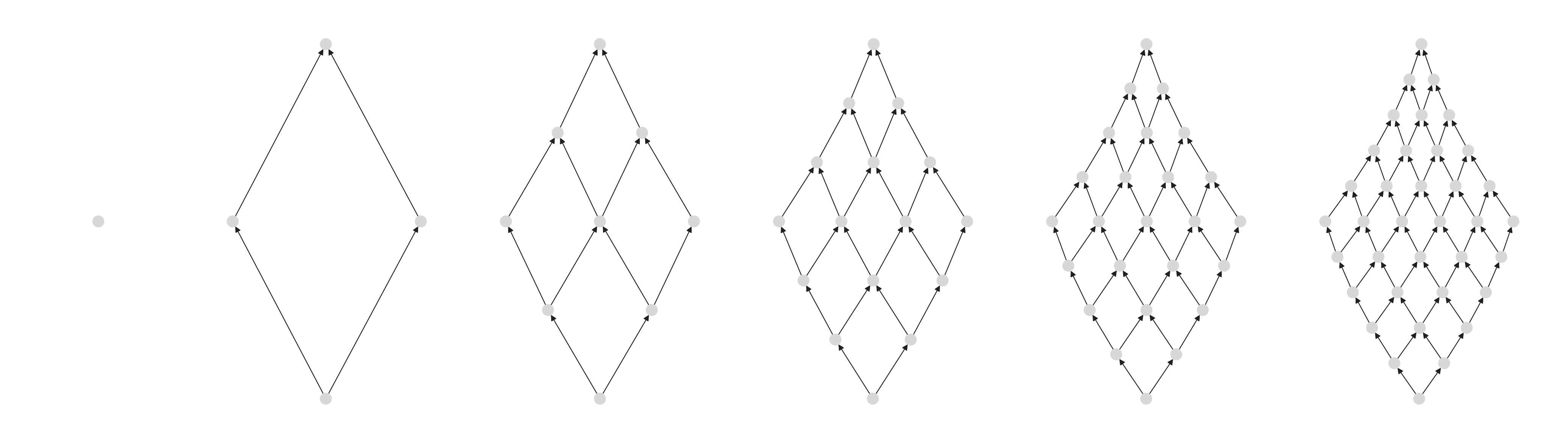}
\caption{Hasse diagrams corresponding to symmetric bicaterpillar trees $\mathcal C_{n,n}$ for $n = 1, 2, 3, 4, 5, 6$. The diagram for $\mathcal C_{1,1}$ consists of a single node and no edges.}
\label{fig:bicat}
\end{figure}
%%%%%%%%%%%%%%%%%%%%%%%%%%%%%%%%%%%%%%%%%%%%%%%%%%%%%%%%%%%%%%%%%%%%%%%%%%%

\clearpage

%%%%%%%%%%%%%%%%%%%%%%%%% Section 8.3 %%%%%%%%%%%%%%%%%%%%%%%%%%%%%%%%%%%%%
\subsection{Lodgepole trees}
\label{subsec:lodg1_ex}

We now consider two families of ``lodgepole'' trees. First, let $\mathcal L_n$ be the family of lodgepole trees defined as in \cite{disanto2015coalescent}. In particular, $\mathcal L_1 = (a,(b,c))$, $\mathcal L_2 = ((a,(b,c)),(e,f))$, $\mathcal L_3 = (((a,(b,c)),(e,f)),(g,h))$, and so on, as shown in Figure~\ref{fig:families4}.

Hasse diagrams for the sets of root configurations for the lodgepole trees appear in Figure~\ref{fig:lodg1}. Using eq.~\ref{eq:recur_count}, we notice that the number of nodes in the Hasse diagrams, or the number of root ancestral configurations, is $c(\mathcal{L}_n)=2^{n+1} - 2$. The number of paths, or number of labeled histories, is the double factorial $h(\mathcal{L}_{n})=(2n-1)!!$, as obtained by eq.~\ref{eq:steel_exp}. The number of root ancestral configurations $c(\mathcal{L}_n)$ and the number of labeled histories $h(\mathcal{L}_{n})$ appear in  Table~\ref{tbl:lodgepole1_numbers}.

%%%%%%%%%%%%%%%%%%% Table 5 %%%%%%%%%%%%%%%%%%%%%%%%%%%
\begin{table}[tb]
\centering
	\begin{tabular}{cccccc}
		\toprule 
            & \multicolumn{5}{c}{$n$} \\ \cline{2-6}
		    & 1 & 2 & 3 & 4 & 5 \\
		\hline
		$c(\mathcal L_{n})$ & 2 & 6 & 14 & 30  & 62 \\
		$h(\mathcal L_{n})$ & 1 & 3 & 15 & 105 & 945 \\
		\bottomrule
	\end{tabular}
	\caption{The number of root ancestral configurations $c(\mathcal L_n)$ and the number of paths $h(\mathcal L_n)$ for lodgepole trees $\mathcal L_n$.}
	\label{tbl:lodgepole1_numbers}
\end{table}
%%%%%%%%%%%%%%%%%%%%%%%%%%%%%%%%%%%%%%%%%%%%%%%%%%%%%%%

Using Theorem~\ref{thm:graph_decomp}, we can define the class of diagrams coming from lodgepole trees as
\begin{equation}\label{eq:decomp_lodg1}
    D(\mathcal L_n) = ((\emptyset \house I)\house \ldots)\house I,
\end{equation}
where the $\widehat \square$-product is applied $n$ times. Because the right root subtree of $\mathcal L_n$ is a caterpillar $\mathcal C_2$, Proposition~\ref{prop:house_with_caterpillar} can be applied to show that 
\[
    h(\mathcal L_n) = (2n-1) \, h(\mathcal L_{n-1}).
\]
An inductive argument then gives a proof of the fact that $h(\mathcal{L}_{n})=(2n-1)!!$ purely in the language of lattice diagrams, without referencing eq.~\ref{eq:steel_exp}.

%%%%%%%%%%%%%%%%%%%%%%%%%%%%%%% Figure 9 %%%%%%%%%%%%%%%%%%%%%%%%%%%%%%%%%%
\begin{figure}[ht]\centering
\includegraphics[width=\linewidth]{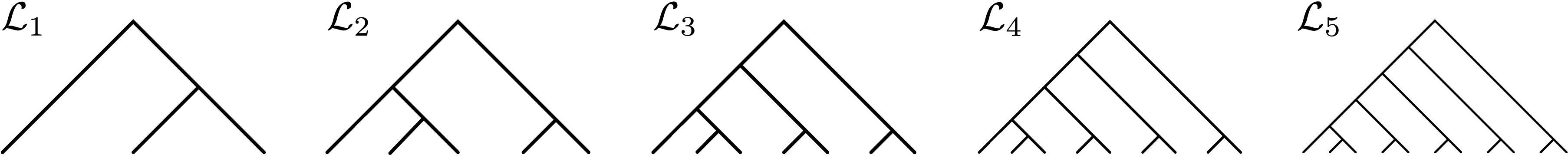}
\caption{Lodgepole trees $\mathcal L_n$ for $n=1,2,3,4,5$.}
\label{fig:families4}
\end{figure}
%%%%%%%%%%%%%%%%%%%%%%%%%%%%%%%%%%%%%%%%%%%%%%%%%%%%%%%%%%%%%%%%%%%%%%%%%%%

%%%%%%%%%%%%%%%%%%%%%%%%%%%%%%% Figure 10 %%%%%%%%%%%%%%%%%%%%%%%%%%%%%%%%%
\begin{figure}[ht]\centering
\includegraphics[width=\linewidth, trim= 200 100 200 100, clip]{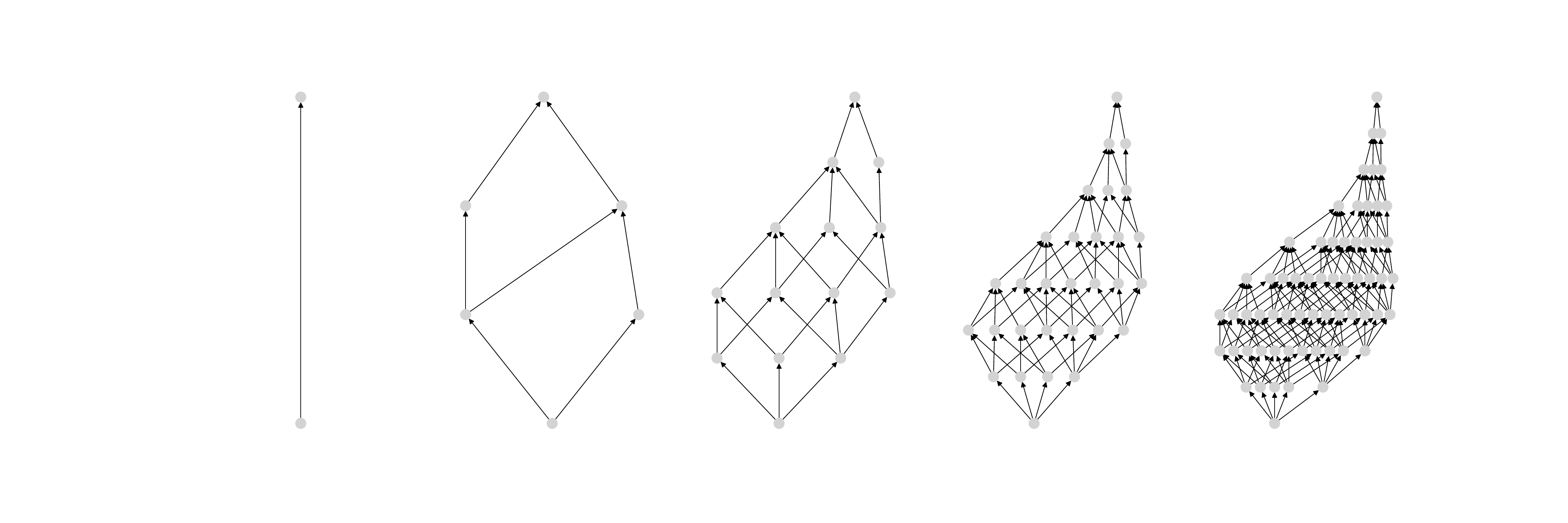}
\caption{Hasse diagrams corresponding to lodgepole trees $\mathcal L_n$ for $n = 1,2,3,4,5$. }
\label{fig:lodg1}
\end{figure}
%%%%%%%%%%%%%%%%%%%%%%%%%%%%%%%%%%%%%%%%%%%%%%%%%%%%%%%%%%%%%%%%%%%%%%%%%%%

\newpage

%%%%%%%%%%%%%%%%%%%%%%%%%%%%%%%%%%%%% Section 8.4 %%%%%%%%%%%%%%%%%%%%%%%%
\subsection{A modified lodgepole family}\label{subsec:lodg2_ex}

For our second family of lodgepole-like trees, let $\mathcal M_n$ be a family akin to lodgepoles, but with a cherry instead of the leftmost leaf. In particular, $\mathcal M_1 = ((a,b),(c,d))$, $\mathcal M_2 = (((a,b),(c,d)),(e,f))$, $\mathcal M_3 = ((((a,b),(c,d)),(e,f)),(g,h))$, and so on, as in Figure~\ref{fig:families5}. 

Hasse diagrams for the sets of root configurations for the modified lodgepole family appear in Figure~\ref{fig:lodg2}. The number of nodes in the diagrams, or the number of root ancestral configurations, as calculated by eq.~\ref{eq:recur_count}, is $h(\mathcal M_n) = 3\times 2^{n} - 2$. The number of paths, or number of labeled histories, is $h(\mathcal M_n) = (2n)!!$ by eq.~\ref{eq:steel_exp}, again a double factorial, but now including only even numbers. We provide the number of root ancestral configurations $c(\mathcal{M}_n)$ and the number of labeled histories $h(\mathcal{M}_{n})$ in Table~\ref{tbl:lodgepole2_numbers}.

%%%%%%%%%%%%%%%%%%% Table 6 %%%%%%%%%%%%%%%%%%%%%%%%%%%
\begin{table}[tb]
\centering
	\begin{tabular}{cccccc}
		\toprule 
            & \multicolumn{5}{c}{$n$} \\ \cline{2-6}
		    & 1 & 2 & 3 & 4 & 5 \\
		\hline
		$c(\mathcal M_{n})$ & 4 & 10 & 22 & 46  & 94 \\
		$h(\mathcal M_{n})$ & 2 & 8  & 48 & 384 & 3840 \\
		\bottomrule
	\end{tabular}
	\caption{The number of root ancestral configurations $c(\mathcal M_n)$ and the number of paths $h(\mathcal M_n)$ for modified lodgepole trees $\mathcal M_n$.}
	\label{tbl:lodgepole2_numbers}
\end{table}
%%%%%%%%%%%%%%%%%%%%%%%%%%%%%%%%%%%%%%%%%%%%%%%%%%%%%%%

Similarly to Section~\ref{subsec:lodg1_ex}, Theorem~\ref{thm:graph_decomp} gives us the following general description of the diagrams for the modified lodgepole family:
\begin{equation}\label{eq:decomp_lodg2}
    D(\mathcal M_n) = ((I \house I)\house \ldots)\house I,
\end{equation}
where the $\widehat \square$-product is applied $n$ times. Analogously, Proposition~\ref{prop:house_with_caterpillar} can be applied to show that 
\[
    h(\mathcal M_n) = (2n) \, h(\mathcal M_{n-1}),
\]
 so that the fact that $h(\mathcal M_{n})=(2n)!!$ is established without reference to eq.~\ref{eq:steel_exp}.

%%%%%%%%%%%%%%%%%%%%%%%%%%%%%% Figure 11 %%%%%%%%%%%%%%%%%%%%%%%%%%%%%%%
\begin{figure}[ht]\centering
\includegraphics[width=\linewidth]{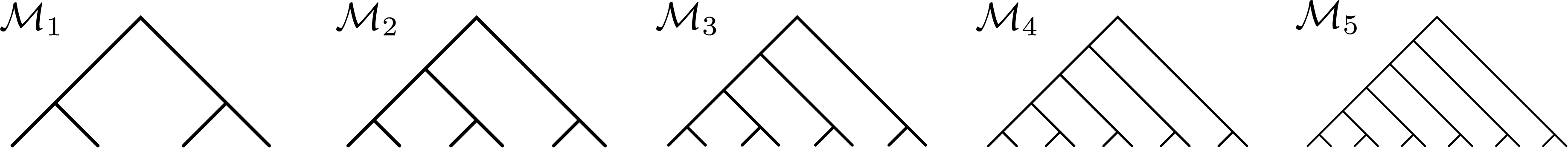}
\caption{Modified lodgepole trees $\mathcal M_n$ for $n=1,2,3,4,5$.}
\label{fig:families5}
\end{figure}
%%%%%%%%%%%%%%%%%%%%%%%%%%%%%%%%%%%%%%%%%%%%%%%%%%%%%%%%%%%%%%%%%%%%%%%%

%%%%%%%%%%%%%%%%%%%%%%%%%%%%%% Figure 12 %%%%%%%%%%%%%%%%%%%%%%%%%%%%%%%
\begin{figure}[ht]\centering
\includegraphics[width=\linewidth, trim= 200 100 200 100, clip]{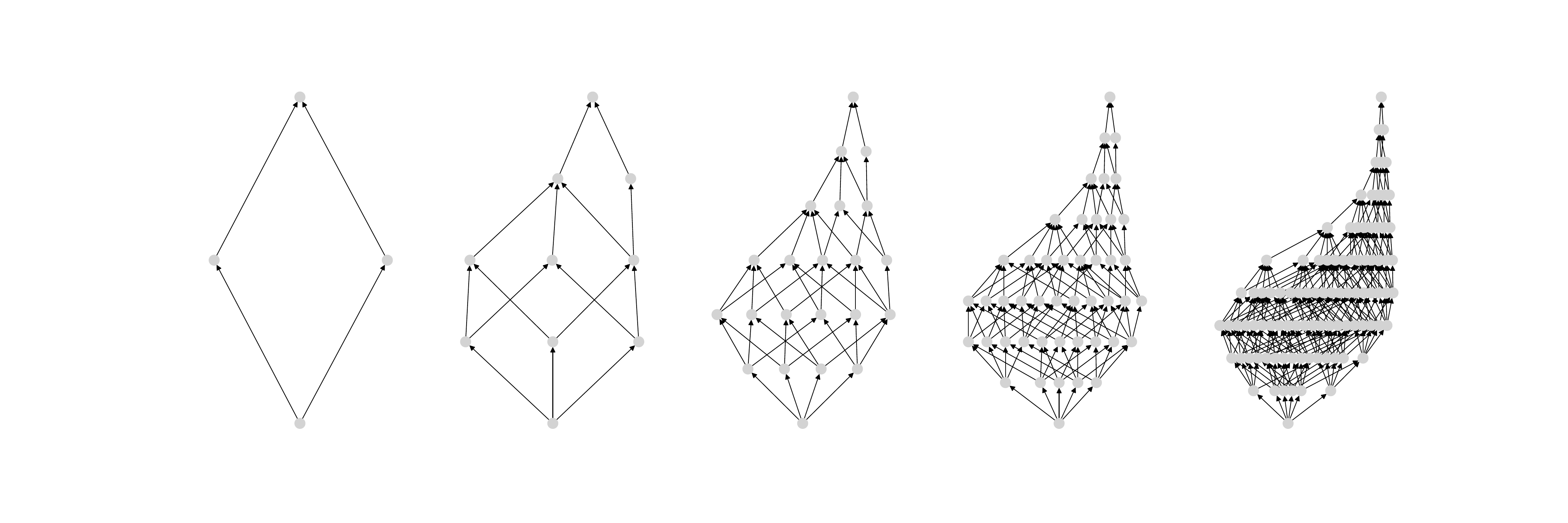}
\caption{Hasse diagrams corresponding to modified lodgepole trees $\mathcal M_n$ for $n = 1,2,3,4,5$.}
\label{fig:lodg2}
\end{figure}
%%%%%%%%%%%%%%%%%%%%%%%%%%%%%%%%%%%%%%%%%%%%%%%%%%%%%%%%%%%%%%%%%%%%%%%%

\newpage

%%%%%%%%%%%%%%%%%%%%%%%%%%%%%%%%%% Section 8.5 %%%%%%%%%%%%%%%%%%%%%%%%%
\subsection{Balanced trees}
\label{subsec:bal_ex}

We define the family of fully balanced trees $\mathcal B_n$ by $\mathcal B_1 = (a,b)$, $\mathcal B_2 = ((a,b),(c,d))$, $\mathcal B_3 = (((a,b),(c,d)),((e,f),(g,h)))$, and so on, where the tree $\mathcal B_n$ has $2^n$ leaves (Figure~\ref{fig:families1}). 

The Hasse diagrams for the sets of root ancestral configurations for $\mathcal B_1$, $\mathcal B_2$, $\mathcal B_3$, and $\mathcal B_4$ appear in Figure~\ref{fig:balanced}. The number of root ancestral configurations for $\mathcal B_n$ is obtained from eq.~\ref{eq:recur_count}, and it is described by a quadratic recurrence $c(\mathcal B_n) = [1 + c(\mathcal B_{n-1})]^2$ with $c(\mathcal B_1)=1$ (A004019 in OEIS). It is also the number of nodes in the diagrams in Figure~\ref{fig:balanced}. Proposition~4 in \cite{disanto2017enumeration} showed that balanced trees $\mathcal B_n$ have the most root ancestral configurations among trees with a fixed size equal to a power of 2; hence, the digraphs $D(\mathcal B_n)$ are the largest digraphs for fixed $n$ when $n$ is a power of two.

%%%%%%%%%%%%%%%%%%% Table 7 %%%%%%%%%%%%%%%%%%%%%%%%%%%
\begin{table}[tb]
\centering
	\begin{tabular}{ccccc}
		\toprule 
            & \multicolumn{4}{c}{$n$} \\ \cline{2-5}
		    & 1 & 2 & 3 & 4 \\
		\hline
		$c(\mathcal B_n)$ & 1 & 4 & 25 & 676 \\
		$h(\mathcal B_n)$ & 1 & 2 & 80 & 21964800 \\
		\bottomrule
	\end{tabular}
	\caption{The number of root ancestral configurations $c(\mathcal{B}_n)$ and the number of paths $h(\mathcal{B}_n)$ for balanced trees $\mathcal{B}_n$.}
	\label{tbl:balanced_numbers}
\end{table}
%%%%%%%%%%%%%%%%%%%%%%%%%%%%%%%%%%%%%%%%%%%%%%%%%%%%%%%%%

The number of labeled histories for  $\mathcal B_n$ is obtained from eq.~\ref{eq:steel_exp}:
\begin{equation}
    h(\mathcal{B}_n) = \frac{(2^n-2)!}{\prod_{k=1}^{n-2} (2^{n-k}-1)^{2^k}}.
\end{equation}
It is the number of paths on the diagrams in Figure~\ref{fig:balanced} (A056972 in OEIS).
We provide the number of root ancestral configurations $c(\mathcal{B}_n)$ and the number of labeled histories $h(\mathcal{B}_n)$ in Table \ref{tbl:balanced_numbers}.

Theorem~\ref{thm:graph_decomp} gives us a recurrent description of the diagrams $D(\mathcal B_n)$:
\[
    D(\mathcal B_1) = I,
\]
\begin{equation} \label{eq:decomp_balances}
    D(\mathcal B_n) = \cphat{D(\mathcal B_{n-1})}{D(\mathcal B_{n-1})}.
\end{equation}
Note that the example of balanced trees shows that the $\widehat \square$-operator is not associative, as 
\[
    (I\house I)\house(I\house I) = D(\mathcal B_3) \neq  D(\mathcal M_3) = ((I\house I)\house I)\house I.
\]

%%%%%%%%%%%%%%%%%%%%%%%%%%%%%% Figure 13 %%%%%%%%%%%%%%%%%%%%%%%%%%%%%%%
\begin{figure}[ht]\centering
\includegraphics[width=0.55\linewidth]{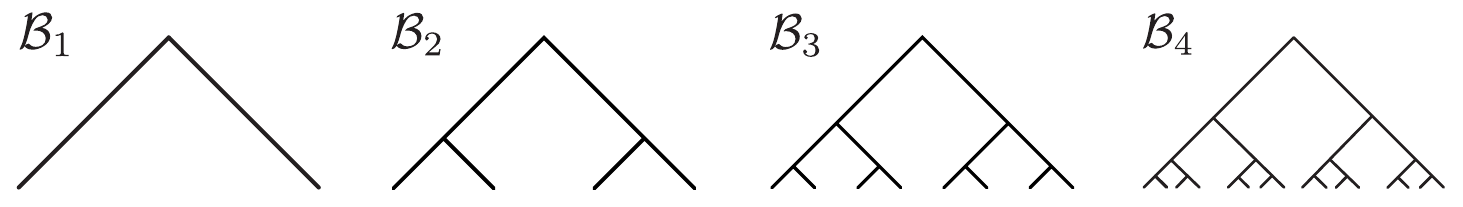}
\caption{Balanced trees $\mathcal B_n$ for $n=1,2,3,4$.}
\label{fig:families1}
\end{figure}
%%%%%%%%%%%%%%%%%%%%%%%%%%%%%%%%%%%%%%%%%%%%%%%%%%%%%%%%%%%%%%%%%%%%%%%

%%%%%%%%%%%%%%%%%%%%%%%%%%%%%%% Figure 14 %%%%%%%%%%%%%%%%%%%%%%%%%%%%%%
\begin{figure}[ht]\centering
\includegraphics[width=0.66\linewidth]{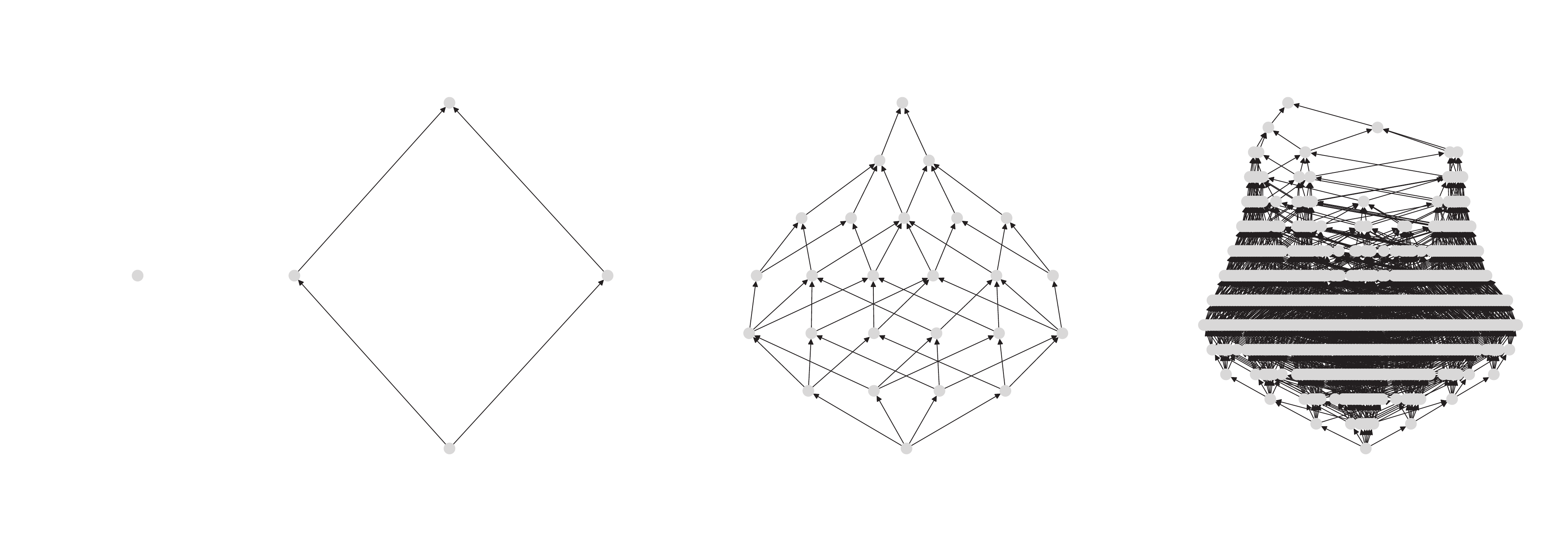}
\caption{Hasse diagrams corresponding to balanced trees $\mathcal B_n$ for $n=1,2,3,4$. The diagram for $\mathcal B_1$ consists of a single node and no edges.}
\label{fig:balanced}
\end{figure}
%%%%%%%%%%%%%%%%%%%%%%%%%%%%%%%%%%%%%%%%%%%%%%%%%%%%%%%%%%%%%%%%%%%%%%%%

\clearpage 

%%%%%%%%%%%%%%%%%%%%%%%%%%%%%%%%%%%%%%% Section 9 %%%%%%%%%%%%%%%%%%%%%%
\section{Nonmatching tree pairs}\label{sec:nonmatching}

We have been considering gene trees and species trees that have the same topology. However, a feature of the lattice construction for ancestral configurations is that the approach can be extended to cases in which the topologies of the gene tree and species tree do not necessarily match.

In biological settings involving tree pairs, the species tree $S$ is generally treated as \emph{fixed} and the gene tree $G$ is treated as \emph{varying}, as a gene tree can be viewed as a random quantity conditional on a species tree~\citep{degnanRosenberg2009}. The species tree is determined by a history of population subdivisions, and gene trees evolve conditionally on that history.

However, when considering lattices of ancestral configurations, the mathematically natural point of view is to instead treat $G$ as \emph{fixed} and to consider its realizations on many different species trees. We will show that in this setting, no ``anomalies'' happen; that is, for a fixed gene tree, the maximal number of ancestral configurations across all species tree topologies is attained by the matching pair. 

For a pair of trees $(G,S)$ with leaf labels bijectively assigned from the same label set, define $C(G,S)$ to be the lattice of root ancestral configurations that appear in realizations of a gene tree $G$ on a species tree $S$. Similarly, we use corresponding notation for other sets and counts we have associated to lattices---reverting to $c(G,S)$, $H(G,S)$, and $h(G,S)$ in place of $c(G)$, $H(G)$, and $h(G)$, and so on. 

We will need the concept of antipodal leaves, following \cite{rosenberg2019lonely}.

\begin{definition}
Consider a pair of trees $(G,S)$ with leaves bijectively labeled by the same label set $X = \{x_1,x_2,\ldots, x_n\}$. $S_\ell$ denotes the ``left'' subtree of the root of $S$, and $S_r$ denotes the ``right'' subtree of the root of $S$. We call two leaves $x_i$, $x_j$ of $G$ \emph{antipodal} if the corresponding leaves in $S$ belong to different subtrees of the root of $S$, with $x_i \in S_\ell$ and $x_j \in S_r$ or vice versa. Similarly, we call two edges $e_1,e_2$ of $G$ \emph{antipodal} if at least one pair of leaves $(v_1, v_2)$, with $v_1$ descended from $e_1$ and $v_2$ descended fom $e_2$, has its two constituent leaves in opposite subtrees of the root of $S$. 
\end{definition}

Now we characterize $C(G,S)$; an example appears in Figure~\ref{fig:nonmatching_example}. First, a trivial proposition states that no ``new'' ancestral configurations arise if we consider a non-matching rather than a matching species tree.
\begin{proposition}\label{prop:nonmatching_are_subsets}
	Consider a pair of trees $(G,S)$ with leaves bijectively labeled by the same label set $X = \{x_1,x_2,\ldots, x_n\}$. As a set, $C(G,S) \subseteq C(G,G)$. 
\end{proposition}

\begin{proof}
    Let $c \in C(G,S)$ be an ancestral configuration, $c = \{\ell_1,\ell_2,\ldots, \ell_k\}$; $c$ is a set of lineages of $G$ that, looking backward in time, can appear right before the root in a realization of gene tree $G$ on species tree $S$. We prove $c \in C(G,G)$ by examining a specifically constructed realization, $R$, of gene tree $G$ on species tree $G$. 

    In the realization $R$ of gene tree $G$ on species tree $S=G$, for each $\ell_i$ in configuration $c$ that consists of a single leaf of $G$, the leaf simply does not engage in any coalescences in $R$. Next, consider each $\ell_i$ in configuration $c$ that consists of an internal node of $G$. In $R$, for each such $\ell_i$, each coalescence required for producing the lineage $\ell_i$ happens immediately ancestral to the corresponding node $\ell_i$ of a species tree $G$, i.e.~immediately ancestral to the most-recent-common-ancestor node of the descendant leaves of $\ell_i$. It is always possible to arrange this, because, with $G$ as a species tree, each lineage of gene tree $G$ with the exception of the root lineage (which does not appear in ancestral configurations) does not contain antipodal pairs among its descendant leaves.

    The constructed realization $R$ displays ancestral configuration $c \in C(G,G)$ at the root. 
\end{proof}

On the level of lattice diagrams, $D(G,S)$ has the following description:
\begin{theorem}\label{thm:nonmatching}
	Let $(G,S)$ be a tree pair. Let $\mathbf c \in C(G,G)$ be the unique ancestral configuration in which all non-antipodal pairs of leaves in $G$, and only those pairs, have coalesced. Then $D(G,S)$ is a subgraph of $D(G,G)$, defined as the subset of edges and vertices that appear in at least one path from a minimal element, the ancestral configuration consisting of all leaf lineages, to $\mathbf c$.
\end{theorem}

It is clear that $\mathbf c$ is unique: by specifying what coalescences of lineages lead to $\mathbf c$ being displayed at the root, we specify what lineages are present in $\mathbf c$.

\begin{proof}
First, applying Proposition~\ref{prop:nonmatching_are_subsets}, we claim that $C(G,S)$ forms a sublattice of $C(G,G)$ or, equivalently, $D(G,S)$ is a subgraph of $D(G,G)$. We can see this by noting that an edge $e$ of $D(G,S)$ corresponds to a coalescence of certain lineages of $G$, so it can be found in $D(G,G)$.

We finish the proof by showing that $D(G,S)$ is precisely the subgraph $D_{\mathbf{c}}$, where $D_{\mathbf c}$ is defined as the subset of edges and vertices that appear in at least one path from a minimal element, the ancestral configuration consisting of all leaf lineages, to $\mathbf c$. We show $D(G,S) \subseteq D_{\mathbf c}$ and $D_{\mathbf c} \subseteq D(G,S)$. 

Each ancestral configuration in $C(G,S)$ contains only lineages that can appear below (that is, ``descended from'') the root of $S$. Hence, each configuration $d$ in $C(G,S)$ precedes $\mathbf c$ with respect to the poset relation, because if $d \not\prec \mathbf c$, then $d$ contains a lineage belonging to an antipodal pair, and such a lineage cannot appear below the species tree root in a realization of $G$ on $S$. Thus, $D(G,S) \subseteq D_{\mathbf c}$.

For the reverse direction, each configuration $d$ with $d \prec \mathbf c$ must contain only lineages produced by coalescences of non-antipodal lineages of $G$. We claim that it therefore is a valid root configuration for $G$ realized in $S$. The proof of this claim is similar to the proof of Proposition~\ref{prop:nonmatching_are_subsets}: define a realization $R$ of $G$ on $S$ by letting each $\ell_i \in d = \{\ell_1, \ell_2, \ldots, \ell_k\}$ happen immediately above (that is, ``ancestral to'') the most recent node of $S$ ancestral to all descendants of $\ell_i$ in $G$. The fact that all lineages are produced by non-antipodal coalescences implies that each $\ell_i$ is placed below the root of $S$ in realization $R$. Hence, $D_{\mathbf c} \subseteq D(G,S)$, completing the proof.
\end{proof}

%%%%%%%%%%%%%%%%%%%%%%%%%%%%%%%%% Figure 15 %%%%%%%%%%%%%%%%%%%%%%%%%%%%%%
\begin{figure}
    \centering
    \includegraphics[width=0.82\textwidth]{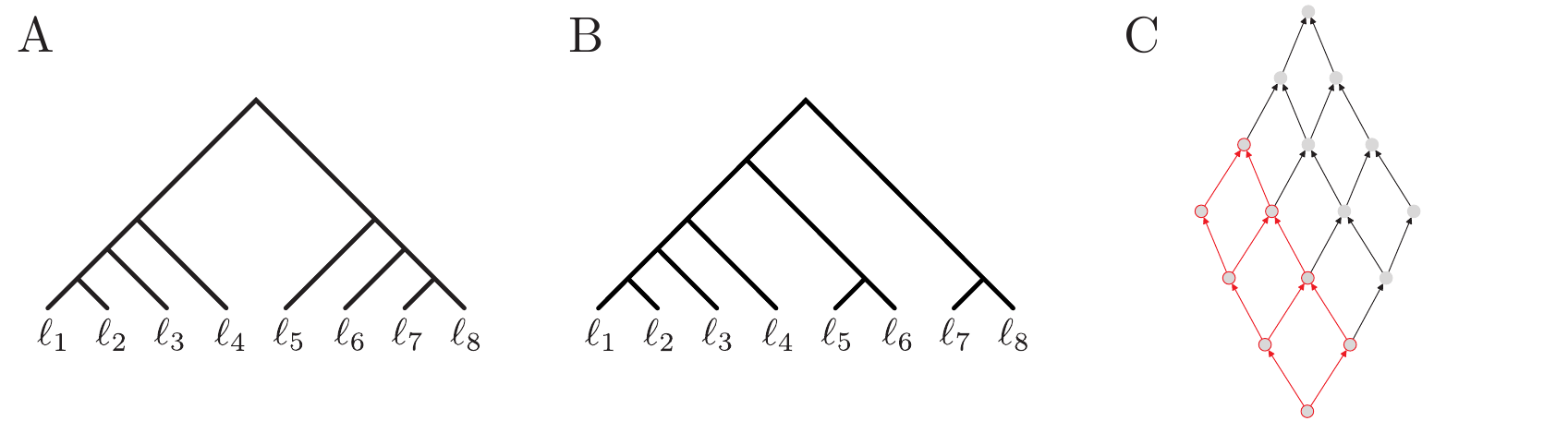}
    \caption{An example lattice of ancestral configurations in the case of nonmatching $G$ and $S$. (A) The gene tree $G = \mathcal C_{4,4}$. (B) A species tree $S$, different from $G$, but having the same leaf labels. (C) Diagrams of ancestral configurations. The digraph $D(G,G)$ is shown in black, as in Figure~\ref{fig:bicat}, and $D(G,S)$ is a subgraph shown in red. We observe that the coalescence $(\ell_7, \ell_8)$ is possible below the root of $S$, while $\ell_6$ and $\ell_7$ form an antipodal pair, forcing $(\ell_6,(\ell_7,\ell_8))$ to happen above the root and the corresponding edges of $D(G,G)$ to not be included in $D(G,S)$.}
    \label{fig:nonmatching_example}
\end{figure}
%%%%%%%%%%%%%%%%%%%%%%%%%%%%%%%%%%%%%%%%%%%%%%%%%%%%%%%%%%%%%%%%%%%%%%%%%%

We now characterize pairs of trees having the maximal number of root ancestral configurations for fixed $G$. Because $c(G,S) \leq c(G,G)$ for all choices of $S$, the trees $S$ for which $c(G,S) = c(G,G)$ are the pairs that have the maximum possible number of root ancestral configurations.
\begin{corollary}
	 The species trees $S$ in pairs $(G,S)$ that for a fixed $G$ have the largest number of root ancestral configurations are exactly the species trees $S$ for which $(G,S)$ has the same set of pairs of antipodal leaves as $(G,G)$. 
\label{coro:12}
\end{corollary}
In other words, decomposing the leaves of a tree at the root into two subsets---those in one subtree of the root and those in the other---the pairs $(G,S)$ that for a fixed $G$ have the largest number of root ancestral configurations are those for which $S$ produces the same decomposition as $G$.
\begin{proof}
    Proposition~\ref{prop:nonmatching_are_subsets} gives $c(G,S) \leq c(G,G)$. Suppose $G$ and $S$ have the same decomposition at the root,
    Theorem~\ref{thm:nonmatching} gives $c(G,S) = c(G,G)$, as the ancestral configuration $\mathbf c$ of $(G,G)$ in which all non-antipodal pairs of lineages have coalesced is the maximal element of both $C(G,G)$ and $C(G,S)$. Hence, $D(G,S) = D(G, G)$. 
    
    Now we show that these are the only $S$ for which $c(G,S) = c(G,G)$. Suppose $c(G,S) = c(G,G)$. Then we must have $D(G,S)=D(G,G)$; because $D(G,S)$ is a sublattice of $D(G,G)$, this equality can hold only if the maximal elements of the two lattices are the same. In other words, $\mathbf c$---the maximal element of $D(G,S)$---is also the maximal element of $D(G,G)$. We conclude that all coalescences of $G$ on $S$ involve only non-antipodal pairs of lineages with respect to $S$, so that $S$ and $G$ possess the same decomposition at the root.
\end{proof}

\begin{example}
	The pair $(G,S) = (\mathcal P_n, \mathcal C_n)$ with $n\geq 5$ satisfies the property in Corollary \ref{coro:12}, and therefore $c(\mathcal P_n, \mathcal C_n) = c(\mathcal P_n, \mathcal P_n)$. 
	\label{examp:13}
\end{example}

Note that although the lattice construction quickly enables the result $c(G,S) \leq c(G,G)$ (Proposition \ref{prop:nonmatching_are_subsets}), and Corollary \ref{coro:12} and Example \ref{examp:13} show that $c(G,S)=c(G,G)$ for some $S \neq G$, it does not follow from these results that $c(G,S) \leq c(S,S)$ for $G \neq S$; indeed, \cite{disanto2017enumeration} reported a counterexample, showing that $c(\mathcal P_n,\mathcal C_n) > c(\mathcal C_n, \mathcal C_n)$ for $n \geq 6$. In general, although the problem of finding $S$ that maximizes $c(G,S)$ is solved, with maximum at $S=G$ and potentially other $S$ as in Corollary \ref{coro:12}, the solution for the problem of finding $G$ that maximizes $c(G,S)$ for fixed $S$ is not currently known.

If \emph{both} $G$ and $S$ are allowed to vary, then the maximal number of ancestral configurations among pairs $(G,S)$ with $n$ leaves is obtained by combining Proposition 4 of \cite{disanto2017enumeration} with our Proposition~\ref{prop:nonmatching_are_subsets}. 
Proposition~\ref{prop:nonmatching_are_subsets} implies that $\max_{G,S} c(G,S) = \max_{G} c(G,G)$. The pair maximizing $c(G,G)$ over all possible $G$ was found by \cite{disanto2017enumeration} it is a certain recursively defined sequence of trees $T_n$. Our Proposition~\ref{prop:nonmatching_are_subsets} implies that $c(T_n,T_n)$ is also maximal among all pairs $(G,S)$.

%%%%%%%%%%%%%%%%%%%%%%%%%%%%%%%% Section 10 %%%%%%%%%%%%%%%%%%%%%%%%%%%%
\section{Discussion}
\label{sec:discussion}

We have demonstrated that the set of ancestral configurations $C(G,S)$ for a rooted binary gene tree $G$ and bijectively labeled species tree $S$ can be described by a lattice structure (Propositions~\ref{prop:acs_are_lattices} and \ref{prop:nonmatching_are_subsets}). For matching gene trees and species trees, we have constructed Hasse diagrams and  digraphs for lattices of ancestral configurations, finding a correspondence between the maximal chains of the lattice $C(S,S)$ and the labeled histories for $S$ (Theorem~\ref{thm:histories_paths}).

\subsection{Summary of results}

For a matching tree pair $(S,S)$, we provided the recursive decomposition of the associated digraph of its ancestral configurations (Theorem~\ref{thm:graph_decomp}). For a series of examples---$p$-pseudocaterpillars (Section \ref{sec:pcat_ex}), bicaterpillars (Section \ref{subsec:bicat_ex}), lodgepole trees (Section \ref{subsec:lodg1_ex}), a modification of lodgepole trees (Section \ref{subsec:lodg2_ex}), and fully balanced trees (Section \ref{subsec:bal_ex})---the recursive decomposition enabled us to describe the lattices associated with ancestral configurations of specific types of trees. For each tree family, we counted the associated numbers of root ancestral configurations and paths from the minimal element to the maximal element of the lattice. The recursive decomposition can be extended to nonmatching tree pairs (Theorem \ref{thm:nonmatching}), suggesting the possibility of extending our numerical evaluations for matching families to a variety of nonmatching cases.

Our work can be seen as an extension of the work of \cite{disanto2017enumeration} to define a recurrence for the \emph{number} (eq.~\ref{eq:recur_count}) of ancestral configurations and for the \emph{set} of ancestral configurations (eq.~\ref{eq:recur_enum}). Here, we have extended the recurrence to \emph{lattices}, as represented by digraphs, by using a cartesian product of graphs. This recurrence provides a quick way of building digraphs of ancestral configurations algorithmically. 

By uncovering a lattice structure for ancestral configurations, we have seen that the root ancestral configurations for a pair of matching trees encode information about labeled histories. Through path-counting on digraphs associated to ancestral configurations, we have obained visual proofs of special cases of the general eq.~\ref{eq:steel_exp} on the number of labeled histories associated with a labeled topology. For example, we proved that the number of labeled histories is $h(\mathcal C_{n,n}) = \binom{2n-2}{n-1}$ for bicaterpillars (Section~\ref{subsec:bicat_ex}), and $h(\mathcal L_n) = (2n-1)!!$ for lodgepole trees (Section~\ref{subsec:lodg1_ex}).  

\subsection{Connections within and beyond mathematical phylogenetics}

The study of ancestral configurations traces to the computation of \cite{wu2012coalescent} to evaluate gene tree probabilities conditional on species trees. The lattice structure facilitates enumerations of ancestral configurations that underlie computational complexity calculations in this evaluation. It further provides an algebraic perspective for understanding the how the complexity depends on tree topologies. Notably, constructions reminiscent of those in Section~\ref{sec:diagrams} have previously been used by \cite{song2006counting}, with a different sense of the term ``ancestral configuration,'' in enumerations that underlie computations related to genealogies with recombination.

Our work describes a connection between ancestral configurations and another key structure of mathematical phylogenetics, namely labeled histories. Under an evolutionary model in which each lineage is equally likely to be the next to bifurcate---the \emph{Yule} or \emph{Yule--Harding} model---the probability of obtaining a specified binary leaf-labeled rooted tree is proportional to its number of labeled histories~\citep{harding1971probabilities, brown1994probabilities, steel2001properties, wiehe2021, King2023}. Labeled histories then provide a convenient structure for evaluating probabilities of tree features under the Yule--Harding model, appearing throughout its associated calculations. By showing that labeled histories represent paths through a lattice whose vertices represent ancestral configurations, our work highlights the connection between the two structures: in some sense, it is enough to compute the set of ancestral configurations to also know the set of labeled histories, and vice versa.

The connection between labeled histories and root configurations emerged from the study of \emph{maximal chains} in the lattice $D(S)$. Approaches in which a bijection is defined between a combinatorial structure of interest, such as labeled histories in our case, and chains in some lattice are common in enumerative combinatorics \citep[Ch.~3]{stanley} and theoretical computer science~\citep[Ch.~14]{garg2015introduction}. Interestingly, lattices have recently appeared in a related mathematical phylogenetic study, in which \cite{palacios2021enumeration} formalized a sense of ``refinement'' for ``perfect phylogenies.'' The lattices of \cite{palacios2021enumeration} are similar to ours, in that in both cases, the ``covering relations,'' corresponding to edges of Hasse diagrams, arise from a ``collapse'' of objects on a tree --- tree edges collapse to a point for \cite{palacios2021enumeration} and lineages collapse to an ancestor in our case. \cite{palacios2021enumeration} used the lattice construction to count certain classes of perfect phylogenies; taken together with that study, our work illustrates that partial orders and lattices have much potential to contribute to enumerative analysis in mathematical phylogenetics.

One promising direction for further application of the lattice construction is in the analysis of phylogenetic networks. Certain classes of phylogenetic network that, like rooted trees, can be temporally embedded to represent biological processes occurring in time \citep{bienvenu2022rtcn, Mathur2023}, possess analogous structures to ancestral configurations: sets that describe the lineages present at a node of a network at a given moment in time. These analogous structures may be possible to understand by studying associated lattices of root configurations for trees displayed by the networks.

Our analyses of particular tree families illustrate the mathematical approach of examining caterpillars and related shapes as relatively simple examples in the study of combinatorial structures involving trees. Caterpillar-like trees can provide the simplest case of recursive tree problems, as one of the subtrees of the root is trivial. Indeed, our caterpillar-based recursive result in Proposition~\ref{prop:house_with_caterpillar} recalls analyses of coalescent histories, a combinatorial structure that, like ancestral configurations, is defined for gene trees and species trees~\citep{degnan2005gene, rosenberg2007counting}. For coalescent histories, caterpillar-like families have been among the first tree shapes to receive extensive analysis~\citep{rosenberg2013coalescent,disanto2015asymptotic,himwich2020,alimpiev2021enumeration}.

Beyond their role in mathematical phylogenetics, ancestral configurations contribute to broader problems in analytic combinatorics and the analysis of algorithms. For example, \cite{disanto2017enumeration} found that the growth of ancestral configurations for balanced families produced an elegant example of the \cite{Aho1973} technique for evaluating asymptotics of quadratic recursions---in a form that was then used in assessing the \cite{Colijn2018} bijection of unlabeled trees with positive integers~\citep{Rosenberg2021}. In \cite{disanto2022distributions}, the asymptotic growth of the coefficients of a generating function in a problem concerning ancestral configurations was derived without the generating function itself---using only a Riccati differential equation satisfied by the generating function. By reframing ancestral configurations in terms of new structures, the algebraic connections of the lattice formulation, ancestral configurations, and labeled histories have potential to inspire further results with broad mathematical connections in combinatorics.

\subsection{Extensions}

Several enumerative problems arise directly from the bijection between paths on digraphs $D(S)$ and labeled histories. For example, how many labeled histories of $S$ are possible such that the realization of gene tree $S$ on species tree $S$ can display root ancestral configuration $c$? In the language of digraphs, each labeled history corresponds to a path, and each ancestral configuration corresponds to a vertex, so that this enumerative problem is equivalent to a problem of counting paths from the minimal to the maximal element that pass through a chosen node. As a path counting problem, the problem is similar to some considered for coalescent histories~\citep{himwich2020, alimpiev2021enumeration}.

Similarly, suppose partial information about the temporal ordering of \emph{some} (possibly all) of the coalescences of $S$ on $S$ is assumed: for example, suppose coalescence $\ell_i$ precedes $\ell_j$. In this case, we ask for the number of root ancestral configurations $c$ such that there exists a realization of $S$ on $S$ that displays $c$ and that satisfies given conditions on the temporal ordering of coalescences. The lattice interpretation of root ancestral configurations then allows us to reinterpret this problem as counting the number of vertices that are encountered by paths in a given subset of paths from the minimal to the maximal element.

\vskip .3cm
\noindent {\small
{\bf Acknowledgments.} We acknowledge support from NIH grant R01 HG005855.
}

{\small
\bibliographystyle{tpb}
\bibliography{sources.bib}
}

\end{document}